\documentclass[twocolumn,10pt]{IEEEtran}
\usepackage{amssymb}
\usepackage{theorem,amsbsy,pifont,verbatim,amsfonts,amssymb,color}
\usepackage{xspace,epsfig,graphicx,subfigure,latexsym,fancyhdr,cite}
\usepackage{multirow}
\usepackage{epstopdf}
\usepackage{makecell}
\usepackage{bm}
\usepackage{booktabs}
\usepackage{indentfirst}
\usepackage{amsmath}
\usepackage{diagbox}
\usepackage{enumerate}
\usepackage{float}
\usepackage{url}
\usepackage{mdwmath}
\usepackage{mdwtab}
\usepackage{xcolor}
\usepackage{caption}
\usepackage{algorithm}
\usepackage{algpseudocode}
\usepackage[colorlinks,citecolor=blue,linkcolor=blue,urlcolor=blue]{hyperref}

\newtheorem{theorem}{Theorem}
\newtheorem{lemma}{Lemma}
\newtheorem{proposition}{Proposition}
\newtheorem{corollary}{Corollary}
\newtheorem{remark}{Remark}
\newtheorem{definition}{Definition}
\newenvironment{proof}{{\it Proof:}}{\hfill $\blacksquare$\par}
\allowdisplaybreaks

\ifCLASSINFOpdf
\else
\fi

\usepackage{stfloats}

\begin{document}

\title{Fundamental Analysis of Scalable Fluid Antenna Systems: Identifiability Limits, Information Theory, and Joint Processing}

\author{Tuo Wu,  Kai-Kit Wong, \emph{Fellow, IEEE}, Jie Tang, Ye Tian,  Baiyang Liu, Maged Elkashlan,   \\Kin-Fai Tong, \emph{Fellow, IEEE},    Hing Cheung So, \emph{Fellow, IEEE}, Matthew C. Valenti, \emph{Fellow, IEEE},\\ ~Fumiyuki Adachi, \IEEEmembership{Life Fellow,~IEEE},    Kwai-Man Luk, \emph{Life Fellow, IEEE}
\thanks{(\textit{Corresponding author: Kai-Kit Wong.})}
\thanks{T. Wu and J. Tang are with the School of Electronic and Information Engineering, South China University of Technology, Guangzhou 510640, China (E-mail: $\rm  \{wutuo,eejtang\}@scut.edu.cn$).      K.-K. Wong is with the Department of Electronic and Electrical Engineering, University College London, WC1E 6BT London, U.K., and also with the Yonsei Frontier Laboratory and the School of Integrated Technology, Yonsei University, Seoul 03722, South Korea (E-mail:$\rm  kai$-$\rm kit.wong@ucl.ac.uk$). Y. Tian is with the Faculty of Electrical Engineering and Computer Science, Ningbo University, Ningbo 315211, China (E-mail: $\rm tianye1@nbu.edu.cn$). B. Liu and  K. F. Tong are with the School of Science and Technology, Hong Kong Metropolitan University, Hong Kong SAR, China. (E-mail: $\rm \{byliu,ktong\}@hkmu.edu.hk$).  M. Elkashlan is with the School of Electronic Engineering and Computer Science at Queen Mary University of London, London E1 4NS, U.K. (E-mail: $\rm maged.elkashlan@qmul.ac.uk$).   H. C. So is with the Department of Electrical  Engineering, City University of Hong Kong, Hong Kong (E-mail: $\rm  hcso@ee.cityu.edu.hk$). M. C. Valenti is with the Lane Department of Computer Science and Electrical Engineering, West Virginia University, Morgantown, USA (E-mail: $\rm valenti@ieee.org$). F. Adachi is with the International Research Institute of Disaster Science (IRIDeS), Tohoku University, Sendai, Japan (E-mail: $\rm adachi@ecei.tohoku.ac.jp$). K.-M. Luk is with the State Key Laboratory of Terahertz and Millimeter Waves, Department of Electronic Engineering, City University of Hong Kong, Hong Kong. (E-mail: $\rm eekmluk@ee.cityu.edu.hk$).  }}

\markboth{IEEE Transactions on ~Vol.~XX, No.~XX, XX~2026}%
{Wu \MakeLowercase{\textit{et al.}}: Fundamental Analysis of Scalable Fluid Antenna Systems}
\maketitle

\begin{abstract}
Unlike fixed-position arrays whose observation entropy budget is static, the scalable fluid antenna system (S-FAS) can dynamically scale its aperture, creating distinct observation spaces with configuration-dependent entropy budgets. This unique reconfigurability demands an information-theoretic foundation that goes beyond classical algebraic identifiability analysis. This paper develops an \textit{observation entropy framework} for S-FAS that provides a unified basis for deriving identifiability limits, diagnosing processing bottlenecks, and guiding system design. Consider an S-FAS with $M$ antennas, where $K$ narrowband sources impinge on the array and mutual coupling of order $p$ is mitigated via central subarray selection. By establishing that each configuration's identifiability is governed by its observation entropy budget $H(\mathbf{Y}_\alpha) \leq M_{\text{eff},\alpha} \log(2\pi e \sigma_y^2)$, where $H(\cdot)$ denotes differential entropy, $\mathbf{Y}_\alpha$ is the observation matrix for configuration $\alpha \in \{c, e\}$, with $c$ and $e$ denoting compressed and extended, respectively, $M_{\text{eff},\alpha}$ is the effective observation dimension, and $\sigma_y^2$ is the observation variance, we derive a complete capacity hierarchy: the compressed configuration supports $K_{\max,c} = M - 2p - 1$ sources, the extended configuration supports $K_{\max,e} = M - 1$ sources regardless of whether far-field or mixed-field parameters are estimated, and joint spatial stacking of both configurations yields the in-principle bound $K_{\max}^{\text{joint}} = (M_c + M_e) - 1$, where $M_c$ and $M_e$ are the effective dimensions of each configuration. Crucially, the entropy framework reveals insights inaccessible to algebraic methods: the data processing inequality explains \textit{why} sequential two-stage processing creates an information bottleneck, limiting the sequential capacity to $K_{\text{seq}} = K_{\max,c}$, and the noise entropy ratio provides a diagnostic tool to distinguish fundamental degrees-of-freedom exhaustion from algorithmic suboptimality. The proposed joint multiple signal classification (J-MUSIC) algorithm exploits augmented steering vectors to approach the joint capacity bound. Comprehensive Monte Carlo simulations with dual validation—algebraic (noise subspace dimension) and information-theoretic (noise entropy ratio)—confirm the predicted boundary behavior and capacity hierarchy across all configurations.
\end{abstract}


\section{Introduction}

\IEEEPARstart{T}{he} sixth-generation wireless networks are expected to support centimeter-level localization accuracy for emerging applications such as autonomous driving, industrial Internet of Things (IoT), and augmented reality \cite{FAS_6G}. Direction-of-arrival (DOA) estimation plays a crucial role in these applications, serving as the foundation for acquiring channel state information (CSI) and performing effective downlink beamforming \cite{DOA_CSI,DOA_positioning}, while also enhancing target detection and tracking capabilities in radar and sonar systems \cite{DOA_radar}. To date, a plethora of excellent DOA estimation methods have been proposed, including subspace-based approaches like multiple signal classification (MUSIC) \cite{music} and estimation of signal parameters via rotational invariance techniques (ESPRIT) \cite{esprit}, sparse signal reconstruction methods such as sparse Bayesian learning (SBL) \cite{sbl}, and deep learning-based solutions \cite{DL_DOA,DL_DOA2}.

However, these established methods are fundamentally built upon fixed-position arrays (FPAs) with inter-element spacing $d$ typically no larger than half the carrier wavelength $\lambda$. Such a rigid architecture inherently suffers from two major limitations. First, it suffers from significant mutual coupling, which severely degrades estimation performance \cite{mutual_coupling}. Second, the static steering vector of an FPA corresponds to a fixed number of degrees-of-freedom (DoFs), restricting the ability to achieve super-resolution and resolve a large number of sources \cite{nested_array}. While massive multiple-input multiple-output (MIMO) arrays can partially address these issues, their high hardware costs and power consumption are often prohibitive \cite{massive_MIMO}.

From a broader information-theoretic viewpoint, DoF is a fundamental limiting resource that governs how reliability and resolvability trade with the observation dimension in multi-antenna systems \cite{zheng_tse_dmt}.

As a promising alternative, the fluid antenna system (FAS) has emerged in recent years to overcome the limitations of FPA systems \cite{FAS_wong,FAS_limits,FAS_tutorial}. In a FAS, the position of each radiating element can be dynamically reconfigured across a spatial aperture, enabling flexible and adaptive control over the antenna's spatial behavior. This reconfigurability can be achieved through various means, such as electronically switchable pixel arrays, metasurfaces, or other tunable structures \cite{FAS_pixel,FAS_MIMO}. By allowing the effective radiation point to change in response to the environment or communication needs, FAS unlocks additional spatial DoFs, offering new opportunities for enhancing performance in next-generation wireless systems. Motivated by these advantages, extensive research has explored FAS-enabled schemes, including fluid antenna multiple access (FAMA) \cite{FAMA,FAMA2}, channel estimation \cite{FAS_channel,FAS_channel2}, beamformer design \cite{FAS_beamforming}, and integrated sensing and communications (ISAC) \cite{FAS_ISAC,FAS_ISAC2}.

Despite the demonstrated versatility of FAS, its potential for enhancing direction-finding capabilities remains largely unexplored. The unique characteristics of FAS offer a key advantage for DOA estimation: the dynamic movement of the antenna can construct a larger virtual array, significantly increasing spatial DoFs; this not only enhances estimation accuracy but also enables underdetermined DOA estimation, where the number of detectable sources exceeds the number of physical antennas. Furthermore, the flexible antenna placement allows for adaptive array configurations optimized for different scenarios, providing superior spatial resolution.

To exploit these advantages, the scalable fluid antenna system (S-FAS) was recently proposed as a new paradigm for array signal processing \cite{scalable_FAS_paper}. Unlike conventional FAS, S-FAS is specifically designed for source localization through dynamic aperture scaling. The core innovation lies in its ability to dynamically \textit{scale} its physical aperture through a software-controlled mechanism, switching between two complementary configurations: a \textit{compressed configuration} with sub-wavelength spacing that eliminates grating lobe ambiguities for robust initial DOA estimation, and an \textit{extended configuration} with half-wavelength spacing that provides enhanced spatial resolution for precise joint angle-range refinement. This two-stage framework achieves high-precision localization across all field regimes—near-field, Fresnel, and far-field—without requiring \textit{a priori} field classification \cite{scalable_FAS_paper,mixed_field}.

Despite the demonstrated empirical success of S-FAS, its fundamental theoretical limits remain unexplored. Critically, characterizing these limits for a \textit{reconfigurable} multi-configuration system like S-FAS requires going beyond classical algebraic identifiability analysis (i.e., rank counting of steering matrices). The reason is that S-FAS introduces a new phenomenon to array signal processing: \textit{dynamic observation entropy scaling}. Let $H(\cdot)$ denote the differential entropy, $\mathbf{Y}$ the observation matrix, and $M_{\text{eff}}$ the effective number of array elements after physical constraints are applied. In a conventional FPA, the observation entropy budget $H(\mathbf{Y}) \leq M_{\text{eff}} \log(2\pi e \sigma_y^2)$ (where $\sigma_y^2$ is the observation variance) is fixed by the static array geometry, so algebraic rank counting suffices to determine identifiability. In S-FAS, however, the reconfigurable aperture creates distinct observation spaces with different entropy budgets, namely, $H(\mathbf{Y}_c)$ for the compressed configuration and $H(\mathbf{Y}_e)$ for the extended configuration, which gives rise to cross-configuration information flow that no single-configuration algebraic analysis can capture.

An information-theoretic framework is therefore indispensable for S-FAS, providing three capabilities unavailable from algebraic methods alone: (i) a \textit{unified entropy metric} for fair identifiability comparison across configurations with fundamentally different geometries, coupling characteristics, and element counts; (ii) a \textit{bottleneck diagnosis} mechanism via the data processing inequality, explaining \textit{why} sequential two-stage processing wastes the extended array's capacity—the refined estimates cannot convey more information than that contained in the original compressed observations, i.e., the mutual information satisfies $I(\boldsymbol{\theta}; \hat{\boldsymbol{\theta}}_c) \leq H(\mathbf{Y}_c)$ where $\boldsymbol{\theta}$ contains the source parameters and $\hat{\boldsymbol{\theta}}_c$ the compressed-stage estimates, regardless of the larger aperture in the second stage; and (iii) a \textit{fundamental-versus-algorithmic distinction} that enables S-FAS designers to determine whether performance degradation reflects exhaustion of observational DoFs (indicated by noise entropy collapse as the number of sources approaches $M_{\text{eff}}$) or suboptimal estimation (indicated by residual entropy that the estimator fails to exploit)—a critical requirement for practical system deployment.

Such information--estimation connections have been studied extensively in the information theory literature, including classical relationships linking mutual information and estimation error in Gaussian models \cite{guo_shamai_verdu_imse,palomar_verdu_gradient,verdu_mismatched}.

Complementarily, non-asymptotic information theory and information-spectrum methods characterize how finite data length can create unavoidable performance gaps relative to asymptotic limits, providing a principled lens for understanding practical saturation phenomena \cite{polyanskiy_poor_verdu_fbl,verdu_han_capacity,han_verdu_output_stats}.

Despite the rich literature on array signal processing \cite{music,esprit,sbl} and tensor methods \cite{kruskal,sidiropoulos,parafac,comon,haardt_nossek}, no prior work has addressed identifiability for reconfigurable multi-configuration systems like S-FAS. Classical identifiability results for uniform linear arrays (ULAs) \cite{stoica_nehorai} establish the well-known bound that the number of identifiable sources must be strictly less than the number of array elements for far-field DOA estimation, but these results assume fixed geometry and do not account for configuration-dependent phenomena such as mutual coupling compensation, central subarray selection, or cross-configuration information flow.

To bridge this gap, this paper develops an \textit{observation entropy framework} for S-FAS and derives the complete identifiability hierarchy from this unified foundation. Let $M$ denote the total number of antenna elements, $K$ the number of impinging sources, and $p$ the mutual coupling order mitigated by central subarray selection (removing $p$ edge elements from each end). Let $K_{\max,c}$ and $K_{\max,e}$ denote the maximum identifiable source counts for the compressed and extended configurations, respectively, and let $M_c = M - 2p$ and $M_e = M$ be their effective observation dimensions. The main contributions are:

\begin{itemize}
\item \textbf{\textit{Observation Entropy Framework:}} We introduce the concept of \textit{observation entropy budget} for reconfigurable arrays: each S-FAS configuration possesses an entropy budget $H(\mathbf{Y}_\alpha) \leq M_{\text{eff},\alpha} \log(2\pi e \sigma_y^2)$ that scales with its effective observation dimension $M_{\text{eff},\alpha}$. The fundamental identifiability constraint, expressed via the mutual information $I(\boldsymbol{\theta}; \mathbf{Y}) \leq H(\mathbf{Y})$, requires the noise subspace to retain at least one dimension ($M_{\text{eff}} - K \geq 1$) to serve as the statistical reference for parameter discrimination. All subsequent results are derived as specializations of this framework to specific S-FAS configurations.

\item \textbf{\textit{Compressed Configuration Bound:}} Applying the entropy framework with $M_{\text{eff},c} = M - 2p$ (after central subarray selection for coupling mitigation), we derive $K_{\max,c} = M - 2p - 1$. The information-theoretic proof reveals that edge removal reduces the entropy budget $H(\mathbf{Y}_c)$ while enabling grating-lobe-free initialization—a principled trade-off between entropy loss and spatial unambiguity.

\item \textbf{\textit{Extended Configuration Bound:}} Under the entropy framework with $M_{\text{eff},e} = M$, we prove $K_{\max,e} = M - 1$ regardless of whether far-field or mixed-field parameters are estimated. The entropy perspective clarifies \textit{why} this holds: the constraint arises from the observational DoFs (governing $H(\mathbf{Y})$), not from the number of parameters per source.

\item \textbf{\textit{Sequential Bottleneck Diagnosis:}} Using the data processing inequality within the entropy framework, we prove that sequential capacity $K_{\text{seq}} = \min(K_{\max,c}, K_{\max,e})$ is limited by the compressed-stage entropy budget. This diagnosis is \textit{uniquely} enabled by the information-theoretic perspective—algebraic analysis can detect the bottleneck but cannot explain its mechanism.

\item \textbf{\textit{Entropy-Expanding Joint Processing:}} We propose spatial stacking as an entropy expansion strategy: combining configurations yields $H(\mathbf{Y}_{\text{joint}}) \propto (M_c + M_e)$, breaking through the sequential bottleneck to achieve the in-principle bound $K_{\max}^{\text{joint}} = (M_c + M_e) - 1$. Practical saturation effects arising from manifold conditioning are analyzed in Remark~\ref{rem:theory_practice_gap}.

\item \textbf{\textit{J-MUSIC Algorithm:}} We develop a practical algorithm exploiting augmented steering vectors from both configurations, with complexity scaling as $\mathcal{O}((M_c+M_e)^3)$.

\item \textbf{\textit{Dual Validation:}} For every theoretical bound, we provide both algebraic validation (noise subspace dimension) and information-theoretic validation (noise entropy ratio), confirming that the entropy framework correctly predicts the identifiability hierarchy across all configurations.
\end{itemize}

\section{System Model and Preliminaries}

\subsection{S-FAS Configuration Parameters}

Consider an S-FAS with $M$ antennas whose inter-element spacing is controlled by a scaling factor $\alpha$:
\begin{equation}\label{eq:scaling_spacing}
d(\alpha) = \alpha \cdot d_0, \quad p_m(\alpha) = (m-1)\alpha d_0, \quad D(\alpha) = (M-1)\alpha d_0
\end{equation}
where $d_0 = \lambda/2$ is the baseline spacing, $p_m(\alpha)$ is the $m$-th element position, and $D(\alpha)$ is the array aperture. We focus on dual-configuration operation with $\alpha \in \{\alpha_c, \alpha_e\}$.

The \textit{compressed configuration} ($\alpha_c = 0.1$) yields spacing $d_c = 0.05\lambda$ and aperture $D_c = (M-1)\alpha_c d_0$. This sub-wavelength spacing eliminates grating lobes but introduces severe mutual coupling modeled by Toeplitz matrix $\mathbf{C}_c \in \mathbb{C}^{M \times M}$. Central subarray selection removes $p$ edge elements from each end, giving effective DoF $M_{\text{eff},c} = M - 2p$.

The \textit{extended configuration} ($\alpha_e = 1.0$) has spacing $d_e = 0.5\lambda$ and aperture $D_e = (M-1)\alpha_e d_0$. Mutual coupling is negligible, yielding $M_{\text{eff},e} = M$ effective elements. The aperture gain $D_e/D_c = 10$ provides enhanced resolution.

Throughout this paper, we adopt $M = 32$ and $p = 3$, yielding $M_{\text{eff},c} = 26$ and $M_{\text{eff},e} = 32$.

\subsection{Signal Model}

Consider $K$ narrowband sources at DOAs $\boldsymbol{\theta} = [\theta_1, \ldots, \theta_K]^T$ and ranges $\mathbf{r} = [r_1, \ldots, r_K]^T$ with uncorrelated signals $\mathbf{s}(t) \in \mathbb{C}^{K}$. The received signal at snapshot $t$ in configuration $\alpha \in \{\alpha_c, \alpha_e\}$ is
\begin{equation}\label{eq:general_signal_model}
\mathbf{y}_\alpha(t) = \mathbf{A}_\alpha(\boldsymbol{\theta}, \mathbf{r})\mathbf{s}(t) + \mathbf{n}_\alpha(t)
\end{equation}
where $\mathbf{A}_\alpha \in \mathbb{C}^{M_{\text{eff},\alpha} \times K}$ is the array manifold and $\mathbf{n}_\alpha(t) \sim \mathcal{CN}(\mathbf{0}, \sigma_n^2\mathbf{I})$.

\subsubsection{Exact Spatial Geometry (ESG)}

For source $k$ at $(\theta_k, r_k)$, the exact distance to element $m$ is
\begin{equation}\label{eq:exact_distance}
r_{m,k}(\alpha) = \sqrt{r_k^2 + p_m^2(\alpha) - 2r_k p_m(\alpha) \sin\theta_k}
\end{equation}
giving the ESG steering vector
\begin{equation}\label{eq:esg_steering}
[\mathbf{a}_k^{\text{ESG}}(\theta_k, r_k, \alpha)]_m = \frac{1}{\sqrt{M}} \frac{r_k}{r_{m,k}(\alpha)} e^{j(2\pi/\lambda)[r_{m,k}(\alpha) - r_k]}
\end{equation}
which is valid for all field regions.

\subsubsection{Compressed Configuration}

Mutual coupling is modeled by Toeplitz matrix $\mathbf{C}_c$. For far-field sources ($r_k \gg 2D_c^2/\lambda$), the steering matrix is
\begin{equation}\label{eq:ff_steering_compressed}
[\mathbf{A}_c^{\text{FF}}(\boldsymbol{\theta})]_{m,k} = e^{j\frac{2\pi}{\lambda}(m-1)d_c\sin\theta_k}
\end{equation}
with Vandermonde structure. Central subarray selection via the selection matrix
\begin{equation}\label{eq:selection_matrix}
\mathbf{F} = [\mathbf{0}_{(M-2p) \times p}, \mathbf{I}_{M-2p}, \mathbf{0}_{(M-2p) \times p}]
\end{equation}
yields the effective manifold
\begin{equation}\label{eq:effective_compressed_manifold}
\mathbf{A}_c(\boldsymbol{\theta}) = \mathbf{F}\mathbf{C}_c \mathbf{A}_c^{\text{FF}}(\boldsymbol{\theta}) \in \mathbb{C}^{(M-2p) \times K}.
\end{equation}
The signal model is then given by $\mathbf{y}_c(t) = \mathbf{A}_c(\boldsymbol{\theta})\mathbf{s}(t) + \mathbf{n}_c(t)$.

\subsubsection{Extended Configuration}

With spacing $d_e = 0.5\lambda$, mutual coupling is negligible. For far-field sources,
\begin{equation}\label{eq:ff_steering_extended}
[\mathbf{A}_e^{\text{FF}}(\boldsymbol{\theta})]_{m,k} = e^{j\frac{2\pi}{\lambda}(m-1)d_e\sin\theta_k}, \quad m = 1,\ldots,M
\end{equation}
supports DOA-only estimation. For mixed-field sources, the ESG model
\begin{equation}\label{eq:esg_steering_extended}
[\mathbf{A}_e^{\text{ESG}}(\boldsymbol{\theta}, \mathbf{r})]_{m,k} = \frac{1}{\sqrt{M}} \frac{r_k}{r_{m,k}(\alpha_e)} e^{j(2\pi/\lambda)[r_{m,k}(\alpha_e) - r_k]}
\end{equation}
enables joint angle-range estimation. The signal model is $\mathbf{y}_e(t) = \mathbf{A}_e(\boldsymbol{\theta}, \mathbf{r})\mathbf{s}(t) + \mathbf{n}_e(t)$.

\subsection{Observation Entropy Framework}
\label{sec:entropy_framework}

Having established the configuration-dependent signal models, we now develop the \textit{observation entropy framework} that serves as the unified theoretical foundation for all identifiability results in this paper. The key insight is that each S-FAS configuration possesses a configuration-dependent \textit{entropy budget} that governs its identifiability capacity. Unlike classical algebraic identifiability analysis, which treats each configuration in isolation via rank analysis, the entropy framework provides a unified metric for cross-configuration comparison, bottleneck diagnosis, and fundamental-versus-algorithmic distinction.

The framework rests on three Pillars, formalized as Lemma~\ref{lem:subspace_decomposition}, Proposition~\ref{prop:fundamental_constraint}, and their information-theoretic proof:
\begin{enumerate}
\item \textbf{Entropy budget:} Each configuration $\alpha$ with effective dimension $M_{\text{eff},\alpha}$ has observation entropy bounded by $H(\mathbf{Y}_\alpha) \leq M_{\text{eff},\alpha} \cdot \log(2\pi e \sigma_y^2)$, and identifiability requires $I(\boldsymbol{\theta}; \mathbf{Y}_\alpha) \leq H(\mathbf{Y}_\alpha)$.
\item \textbf{Noise reference requirement:} The noise subspace must retain at least one dimension ($M_{\text{eff}} - K \geq 1$) to provide the statistical reference for parameter discrimination, yielding the universal bound $K_{\max} = M_{\text{eff}} - 1$.
\item \textbf{Entropy hierarchy:} When different configurations are available, their entropy budgets can be compared ($H(\mathbf{Y}_c)$ vs.\ $H(\mathbf{Y}_e)$), cascaded (sequential processing, governed by the data processing inequality), or combined (joint processing, yielding $H(\mathbf{Y}_{\text{joint}}) \propto M_c + M_e$).
\end{enumerate}
The remainder of this subsection formalizes Pillars 1 and 2, while Pillar 3 is developed in Sections~\ref{sec:sfas_capacity}--\ref{sec:joint_capacity}. Table~\ref{tab:framework_roadmap} summarizes how each subsequent section specializes the framework.

\begin{table}[!t]
\centering
\caption{Observation Entropy Framework: Specializes Roadmap}
\label{tab:framework_roadmap}
\renewcommand{\arraystretch}{1.2}
\begin{tabular}{l c c}
\hline
\textbf{Section} & $M_{\text{eff}}$ & \textbf{Framework Mechanism} \\
\hline
III (Compressed) & $M - 2p$ & Entropy budget reduction \\
IV (Extended) & $M$ & Full entropy budget \\
V-A (Sequential) & $\min(M_c, M_e)$ & Data processing inequality \\
V-B (Joint) & $M_c + M_e$ & Entropy expansion \\
\hline
\end{tabular}
\end{table}

For an array configuration with $M_{\text{eff}}$ effective spatial observations (after coupling mitigation and edge removal), the array collects $N$ temporal snapshots $\{\mathbf{y}(t)\}_{t=1}^N$. The primary statistical quantity for subspace-based estimation is the sample covariance matrix
\begin{equation}\label{eq:sample_covariance}
\hat{\mathbf{R}} = \frac{1}{N}\sum_{t=1}^N \mathbf{y}(t)\mathbf{y}^H(t) \in \mathbb{C}^{M_{\text{eff}} \times M_{\text{eff}}}
\end{equation}
which, under the assumption of ergodic source signals and sufficiently large $N$, converges to the theoretical covariance matrix
\begin{equation}\label{eq:theoretical_covariance}
\mathbf{R} = E[\mathbf{y}(t)\mathbf{y}^H(t)] = \mathbf{A}\mathbf{R}_s\mathbf{A}^H + \sigma_n^2\mathbf{I}_{M_{\text{eff}}}
\end{equation}
where $\mathbf{R}_s = E[\mathbf{s}(t)\mathbf{s}^H(t)]$ is the source covariance matrix and we have suppressed configuration subscripts for notational clarity. For uncorrelated sources with powers $P_1, \ldots, P_K$, the source covariance is diagonal: $\mathbf{R}_s = \text{diag}(P_1, \ldots, P_K)$.

\begin{definition} 
The spatial DoF of an array configuration is defined as the rank of the covariance matrix observation space, which equals the effective number of array elements $M_{\text{eff}}$ after all physical constraints (mutual coupling mitigation, edge removal) are applied.
\end{definition}

The DoF concept is intimately connected to the eigenstructure of the covariance matrix. Under the signal model \eqref{eq:general_signal_model} with $K$ uncorrelated sources and full-rank array manifold $\mathbf{A}$, the covariance matrix $\mathbf{R}$ admits an eigenvalue decomposition (EVD) that partitions the $M_{\text{eff}}$-dimensional observation space into orthogonal signal and noise subspaces. To establish this rigorously, we prove the following foundational result:

\begin{lemma} 
\label{lem:subspace_decomposition}
If the array manifold $\mathbf{A} \in \mathbb{C}^{M_{\text{eff}} \times K}$ has full column rank (i.e., $\text{rank}(\mathbf{A}) = K$) and $K < M_{\text{eff}}$, then the covariance matrix $\mathbf{R}$ has exactly $K$ eigenvalues larger than $\sigma_n^2$ and $M_{\text{eff}} - K$ eigenvalues equal to $\sigma_n^2$. The corresponding eigenvectors span orthogonal signal and noise subspaces of dimensions $K$ and $M_{\text{eff}} - K$, respectively.
\end{lemma}

\begin{proof}
The covariance matrix can be rewritten as
\begin{equation}
\mathbf{R} = \mathbf{A}\mathbf{R}_s\mathbf{A}^H + \sigma_n^2\mathbf{I}.
\end{equation}
Since $\mathbf{R}_s$ is positive definite (all source powers $P_k > 0$) and $\mathbf{A}$ has full column rank, the matrix $\mathbf{A}\mathbf{R}_s\mathbf{A}^H$ is positive semidefinite with rank equal to $\text{rank}(\mathbf{A}) = K$. By the eigenvalue perturbation theorem, the eigenvalues of $\mathbf{R}$ consist of $K$ eigenvalues of the form $\sigma_n^2 + \lambda_k$ where $\lambda_k > 0$ are the eigenvalues of $\mathbf{A}\mathbf{R}_s\mathbf{A}^H$, and $M_{\text{eff}} - K$ eigenvalues equal to $\sigma_n^2$. 

Let $\mathbf{U}_s = [\mathbf{u}_1, \ldots, \mathbf{u}_K]$ denote the eigenvectors corresponding to the $K$ largest eigenvalues (signal subspace), and $\mathbf{U}_n = [\mathbf{u}_{K+1}, \ldots, \mathbf{u}_{M_{\text{eff}}}]$ denote the eigenvectors corresponding to the noise eigenvalues (noise subspace). These satisfy the orthogonality relations:
\begin{equation}
\mathbf{U}_s^H\mathbf{U}_n = \mathbf{0}, \quad \mathbf{U}_s^H\mathbf{U}_s = \mathbf{I}_K, \quad \mathbf{U}_n^H\mathbf{U}_n = \mathbf{I}_{M_{\text{eff}}-K}.
\end{equation}
Furthermore, the noise subspace is orthogonal to the array manifold: $\mathbf{A}^H\mathbf{U}_n = \mathbf{0}$, which is the fundamental orthogonality property exploited by MUSIC.
\end{proof}

The practical implication of Lemma~\ref{lem:subspace_decomposition} is that subspace-based methods require at least one dimension in the noise subspace for source parameter estimation. This imposes the fundamental identifiability constraint:

\begin{proposition} 
\label{prop:fundamental_constraint}
For an array configuration with spatial DoFs being $M_{\text{eff}}$, the maximum number of uniquely identifiable sources using subspace methods is
\begin{equation}\label{eq:fundamental_constraint}
K_{\max} = M_{\text{eff}} - 1.
\end{equation}
This bound is tight when the array manifold has full column rank for all parameter combinations.
\end{proposition}

\begin{proof}
We provide a dual proof from both algebraic and information-theoretic perspectives.

\textit{Algebraic proof:} From Lemma~\ref{lem:subspace_decomposition}, we require $M_{\text{eff}} - K \geq 1$ to ensure the existence of a non-trivial noise subspace. This yields $K \leq M_{\text{eff}} - 1$. The bound is achieved when $K = M_{\text{eff}} - 1$, which leaves a one-dimensional noise subspace that is still sufficient for MUSIC spectral search via the orthogonality condition $\mathbf{a}^H(\theta)\mathbf{U}_n = 0$.

To corroborate this algebraic proof, we conduct Monte Carlo simulations ($M_{\text{eff}} = 32$, signal-to-noise ratio (SNR) = 10~dB, $N = 500$) that directly measure the noise subspace dimensions $\dim(\mathcal{N})$ for varying $K \in \{2, 4, \ldots, 34\}$. For each $K$, we form the sample covariance matrix and compute the empirical noise subspace dimension using the true source count, i.e., $\dim(\mathcal{N}) = M_{\text{eff}} - K$, providing a purely geometric verification of the rank relation.

Fig.~\ref{fig:validate_fundamental} shows that the measured noise dimension perfectly coincides with the theoretical line $M_{\text{eff}} - K$ for all $K$, confirming that the covariance rank decomposition holds exactly under realistic SNR and snapshot conditions. The vertical lines mark critical transitions: at $K = 31$, the noise subspace becomes one-dimensional, while at $K = 32$ it vanishes completely, leaving no orthogonality reference $\mathbf{a}^H(\theta)\mathbf{U}_n = 0$ and thereby enforcing the identifiability limit $K_{\max} = M_{\text{eff}} - 1$.

\begin{figure}[!t]
	\centering
	\includegraphics[width=0.48\textwidth]{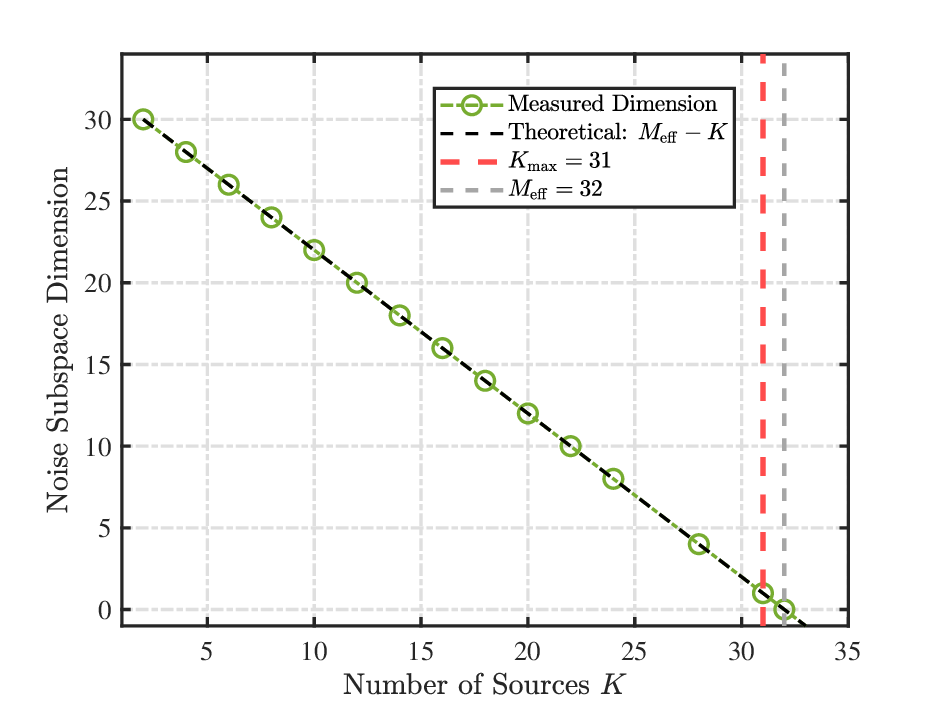}
	\caption{Fundamental constraint: Algebraic validation via $\dim(\mathbf{U}_n)$ versus number of Sources $K$, illustrating Subspace Dimension behavior.}
	\label{fig:validate_fundamental}
\end{figure}

While the algebraic proof establishes the constraint through geometric subspace decomposition, we now provide an independent information-theoretic justification to reveal the fundamental entropy bottleneck underlying this identifiability limit. This dual perspective is crucial: the algebraic view explains \textit{how} subspace methods fail (loss of orthogonality reference), while the information-theoretic view explains \textit{why} parameter estimation becomes impossible (degenerate noise entropy reference).

\textit{Information-theoretic proof:} Consider the mutual information between source parameters $\boldsymbol{\Theta} = \{\theta_1, \ldots, \theta_K\}$ and observations $\mathbf{Y} = \{\mathbf{y}(1), \ldots, \mathbf{y}(N)\}$. The mutual information quantifies the amount of information about source parameters extractable from observations:
\begin{equation}
I(\boldsymbol{\Theta}; \mathbf{Y}) = H(\mathbf{Y}) - H(\mathbf{Y}|\boldsymbol{\Theta})
\end{equation}
where $H(\cdot)$ denotes differential entropy. The observation entropy is upper-bounded by the dimensionality of the observation space:
\begin{equation}
H(\mathbf{Y}) \leq M_{\text{eff}} \cdot \log(2\pi e \sigma_y^2)
\end{equation}
where $\sigma_y^2$ is the total observation variance. Given the source parameters, the conditional entropy reduces to the noise entropy:
\begin{equation}
H(\mathbf{Y}|\boldsymbol{\Theta}) = M_{\text{eff}} \cdot \log(2\pi e \sigma_n^2).
\end{equation}
Thus, the mutual information becomes
\begin{equation}
I(\boldsymbol{\Theta}; \mathbf{Y}) = M_{\text{eff}} \cdot \log\left(1 + \frac{P_{\text{signal}}}{\sigma_n^2}\right) = M_{\text{eff}} \cdot \log(1 + \text{SNR}).
\end{equation}

However, subspace-based estimation fundamentally partitions the observation space into signal and noise subspaces of dimensions $K$ and $M_{\text{eff}} - K$, respectively. The eigenvalue entropy decomposition yields
\begin{align}
H(\boldsymbol{\Lambda}) &= \sum_{i=1}^K \log(\lambda_i) + \sum_{i=K+1}^{M_{\text{eff}}} \log(\sigma_n^2) \nonumber\\
&= H_{\text{signal}}(K) + (M_{\text{eff}} - K)\log(\sigma_n^2)
\end{align}
where $\boldsymbol{\Lambda} = \text{diag}(\lambda_1, \ldots, \lambda_{M_{\text{eff}}})$ contains the covariance eigenvalues. For the noise subspace to provide a non-degenerate reference for parameter estimation, we require
\begin{equation}
H_{\text{noise}} = (M_{\text{eff}} - K)\log(\sigma_n^2) > -\infty \quad \Longrightarrow \quad M_{\text{eff}} - K \geq 1.
\end{equation}
This information-theoretic constraint, requiring at least one dimension to quantify the noise baseline entropy, independently confirms $K \leq M_{\text{eff}} - 1$.
\end{proof}

\begin{remark}[Information-Theoretic Interpretation]
The constraint $K_{\max} = M_{\text{eff}} - 1$ reflects a fundamental information bottleneck: the observation space must allocate at least one dimension to characterize the noise statistics, which serve as the reference baseline for discriminating signal subspace components. When $K = M_{\text{eff}}$, the noise subspace vanishes, causing the noise entropy term to degenerate and eliminating the statistical reference required for parameter identifiability. This is analogous to the Nyquist sampling theorem requiring oversampling by a factor of two; here, the spatial domain requires one redundant dimension for reliable parameter extraction.
\end{remark}

To empirically validate this information-theoretic viewpoint, we perform Monte Carlo simulations with $M_{\text{eff}} = 32$, $N = 500$ snapshots, and different SNR levels ($0$, $10$, $15$~dB). For each source number $K$, we form the sample covariance matrix, compute its eigenvalues, and decompose the observation entropy into signal and noise contributions $H_{\text{signal}}(K)$ and $H_{\text{noise}}(K)$. The resulting noise entropy ratio $\rho_{\text{n}}(K) = |H_{\text{noise}}|/(|H_{\text{signal}}| + |H_{\text{noise}}|)$ quantifies the fraction of observation entropy allocated to the noise baseline.

Fig.~\ref{fig:validate_mutual_information} shows that the noise entropy ratio decays monotonically with the number of sources and collapses to (almost) zero as $K$ reaches $M_{\text{eff}}$ for all SNR levels tested. Once the noise subspace vanishes, no entropy can be allocated to the noise baseline, so the observation space loses the statistical reference required to distinguish signal components from noise, thereby enforcing the identifiability limit $K_{\max} = M_{\text{eff}} - 1$ from a fundamental information-theoretic perspective.

\begin{figure}[!t]
\centering
\includegraphics[width=0.52\textwidth]{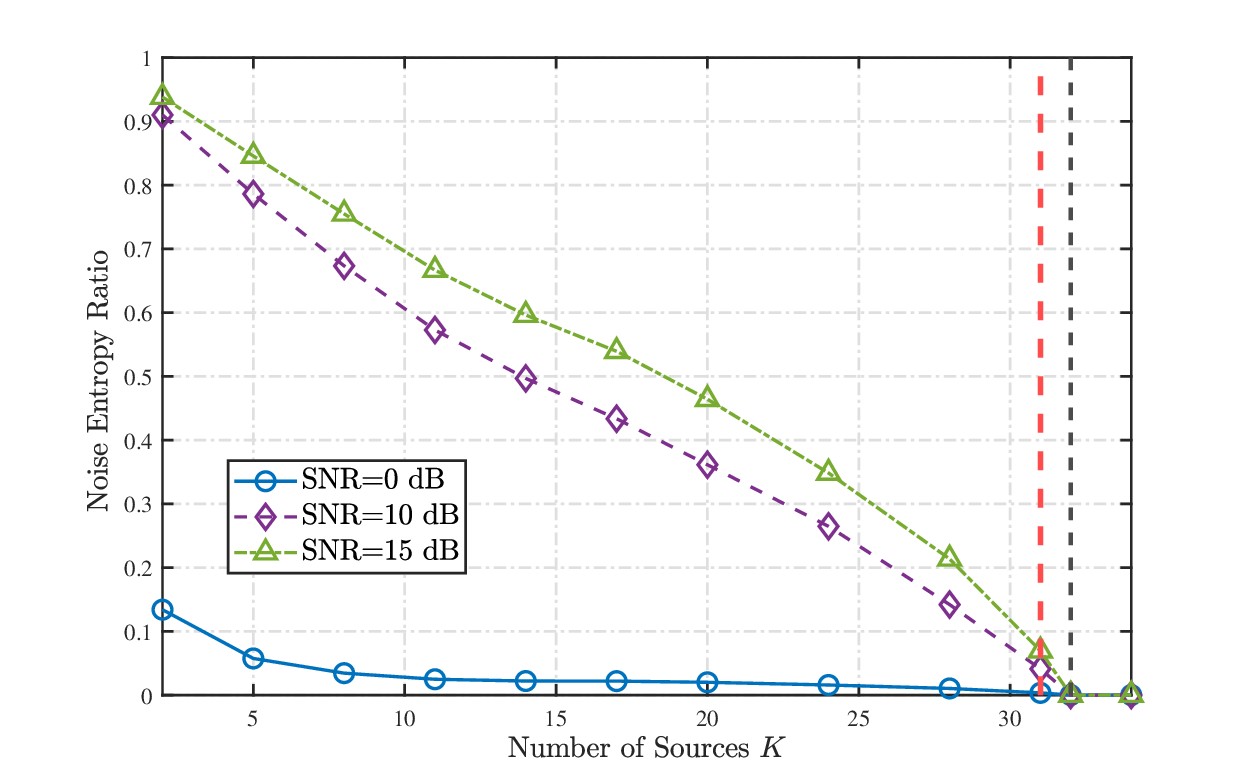}
\caption{Fundamental constraint: Information-theoretic validation via noise Entropy ratio $\rho_n(K)$ at different SNR levels for different numbers of Sources $K$.}
\label{fig:validate_mutual_information}
\end{figure}

Taken together, Proposition~\ref{prop:fundamental_constraint} and Figs.~\ref{fig:validate_fundamental}--\ref{fig:validate_mutual_information} complete the formalization of Pillars~1 and~2 of the observation entropy framework (Section~\ref{sec:entropy_framework}). The global bound $K_{\max} = M_{\text{eff}} - 1$ now serves as the master constraint from which all configuration-specific results are derived: each subsequent section specializes the framework by substituting the appropriate $M_{\text{eff}}$ and applying the corresponding entropy mechanism (see Table~\ref{tab:framework_roadmap}).

\section{Identifiability Analysis: Compressed Configuration}

We now specialize the observation entropy framework (Section~\ref{sec:entropy_framework}) to the compressed S-FAS configuration, applying the \textit{entropy budget reduction} mechanism (Table~\ref{tab:framework_roadmap}): central subarray selection reduces $M_{\text{eff}}$ from $M$ to $M - 2p$, shrinking the entropy budget $H(\mathbf{Y}_c)$ but enabling grating-lobe-free operation. We derive the precise identifiability bound for this configuration. The compressed mode operates with sub-wavelength inter-element spacing $d_c = \alpha_c d_0 < \lambda/2$ to eliminate grating lobes, but this dense packing introduces severe mutual coupling that couples the signals received at neighboring elements. To mitigate coupling effects while preserving spatial information, the S-FAS employs central subarray selection, removing $p$ edge elements from each end where coupling effects are strongest. The key question is: \textit{how many sources can be uniquely identified under these constraints?}

\subsection{Effective DoFs under Coupling Mitigation}

The first step in answering this question is to determine the effective DoFs available after coupling compensation. The sub-wavelength spacing $d_c = \alpha_c d_0$ (with $\alpha_c \ll 1$) induces severe mutual coupling between adjacent elements, modeled by a Toeplitz matrix $\mathbf{C}_c$ with exponentially decaying entries:
\begin{equation}\label{eq:coupling_decay}
c_\ell = c_0 e^{-\beta \ell d_c / \lambda} e^{j\phi_\ell}, \quad |c_1| \approx 0.7
\end{equation}
where $c_\ell$ represents the coupling coefficient between elements separated by $\ell$ positions, $c_0$ is the self-coupling normalization, $\beta$ is the decay rate, and $\phi_\ell$ is the inter-element phase shift, all computed from mutual impedance relationships~\cite{mutual_coupling}. For the sub-wavelength spacing $d_c = 0.1\lambda$ in the compressed configuration, the adjacent-element coupling coefficient is $|c_1| \approx 0.7$, a value extensively validated through electromagnetic analysis and experimental measurements~\cite{nested_array}. To suppress these coupling effects, we apply central sub-array selection via the selection matrix $\mathbf{F}$, which extracts only $M - 2p$ elements while discarding the $p$ edge elements on each end. This strategic removal sacrifices some spatial samples but yields a cleaner effective array manifold. The resulting effective DoFs are characterized by the following lemma:

\begin{lemma} 
\label{lem:compressed_dof}
After central subarray selection with $p$ elements removed from each end, the compressed configuration provides effective spatial DoFs, namely,
\begin{equation}\label{eq:compressed_dof}
\text{DoF}_c = M - 2p.
\end{equation}
\end{lemma}

\begin{proof}
The effective array manifold after coupling and selection is $\mathbf{A}_c(\boldsymbol{\theta}) = \mathbf{F}\mathbf{C}_c\mathbf{A}_c^{\text{FF}}(\boldsymbol{\theta})$ from \eqref{eq:effective_compressed_manifold}. To determine the DoF, we must establish the rank of this composition. 

First, observe that $\mathbf{F} \in \mathbb{R}^{(M-2p) \times M}$ has full row rank since it simply extracts a contiguous subset of rows (the central $M-2p$ rows). Next, the mutual coupling matrix $\mathbf{C}_c$ is Toeplitz with non-zero diagonal entries (representing self-coupling), and under physically realistic coupling models satisfying \eqref{eq:coupling_decay}, it is invertible and thus has full rank $M$. The far-field steering matrix $\mathbf{A}_c^{\text{FF}}(\boldsymbol{\theta}) \in \mathbb{C}^{M \times K}$ has Vandermonde structure with full column rank $K$ when all DOAs are distinct.

Applying the rank inequality for matrix products, we have
\begin{equation}
\text{rank}(\mathbf{F}\mathbf{C}_c\mathbf{A}_c^{\text{FF}}) \geq \text{rank}(\mathbf{F}\mathbf{C}_c) + \text{rank}(\mathbf{A}_c^{\text{FF}}) - M.
\end{equation}
Since $\mathbf{C}_c$ is full-rank and $\mathbf{F}$ has full row rank, we obtain $\text{rank}(\mathbf{F}\mathbf{C}_c) = M - 2p$. For $K \leq M - 2p$, the effective manifold $\mathbf{A}_c(\boldsymbol{\theta})$ has dimensions $(M-2p) \times K$ with column rank at most $M - 2p$. The observation space is therefore $(M-2p)$-dimensional, establishing the effective DoFs as $M - 2p$, namely, $\text{DoF}_c = M - 2p$.
\end{proof}

To empirically validate Lemma~\ref{lem:compressed_dof}, we perform Monte Carlo simulations on an ideal uniform linear array with $M = 32$ sensors, half-wavelength spacing, no mutual coupling, SNR = 10~dB, and $N = 200$ snapshots. For each edge-removal index $p \in \{0, \ldots, 6\}$, we retain the central $M_c = M - 2p$ elements, generate $K$ far-field sources with distinct DOAs, and form the effective covariance matrix after selection. 

Fig.~\ref{fig:dof_reduction} shows that the simulated maximum identifiable sources perfectly matches the theoretical bound $K_{\max} = M_c - 1 = M - 2p - 1$ across all $p$ values, with near-zero deviation. The effective array size $M_c = M - 2p$ decreases linearly, illustrating that each pair of the removed edge elements sacrifices exactly two spatial DoFs and one identifiable source. In particular, for the baseline choice $M = 32$ and $p = 3$, we obtain $M_c = 26$ and $K_{\max} = 25$, which will serve as the compressed configuration benchmark in subsequent analysis.

\begin{figure}[!t]
\centering
\includegraphics[width=0.48\textwidth]{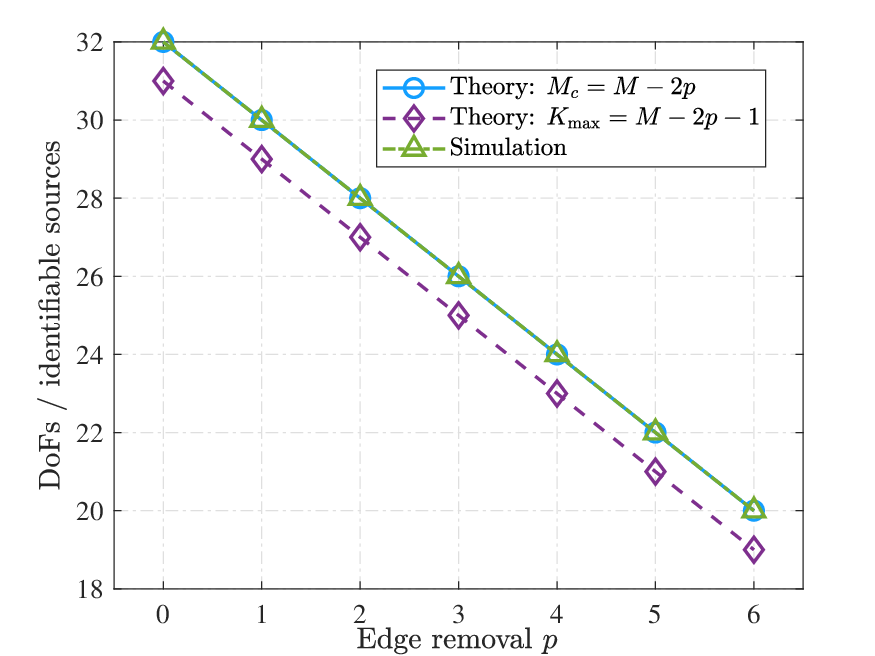}
\caption{DoF Reduction: Validation showing $K_{\max} = M - 2p - 1$ as function of edge Removal $p$ for different values of $p$.}
\label{fig:dof_reduction}
\end{figure}

\subsection{Identifiability Bound}

Applying the fundamental identifiability constraint from Proposition~\ref{prop:fundamental_constraint}, the compressed configuration requires at least one noise subspace dimension for reliable MUSIC-based estimation. This constraint, combined with the effective DoF, namely $\text{DoF}_c = M - 2p$, directly yields the identifiability bound:

\begin{theorem} 
\label{thm:compressed_identifiability}
For the compressed configuration with far-field approximation, central subarray selection, and subspace-based estimation, the maximum number of uniquely identifiable sources is
\begin{equation}\label{eq:kmax_compressed}
K_{\max,c}^{\text{FF}} = M - 2p - 1
\end{equation}
provided the following conditions hold:
\begin{enumerate}
\item \textit{Angular separability:} The source DOAs satisfy $|\theta_i - \theta_j| \geq \Delta\theta_{\min}$ for all $i \neq j$, where $\Delta\theta_{\min}$ is the angular resolution limit.
\item \textit{Full column rank:} The effective array manifold $\mathbf{A}_c(\boldsymbol{\theta}) = \mathbf{F}\mathbf{C}_c\mathbf{A}_c^{\text{FF}}(\boldsymbol{\theta})$ has rank equal to $K$.
\end{enumerate}
\end{theorem}

\begin{proof}
From Proposition~\ref{prop:fundamental_constraint}, subspace methods require at least a one-dimensional noise subspace for parameter estimation via the orthogonality principle. Applying this fundamental constraint to the compressed configuration with $\text{DoF}_c = M - 2p$ effective spatial samples yields
\begin{equation}
K < M - 2p \quad \Longrightarrow \quad K \leq M - 2p - 1.
\end{equation}

To verify that this bound is achievable, we must establish that the array manifold has full column rank for $K = M - 2p - 1$ sources. Consider the manifold structure:
\begin{equation}
\mathbf{A}_c(\boldsymbol{\theta}) = \mathbf{F}\mathbf{C}_c\mathbf{A}_c^{\text{FF}}(\boldsymbol{\theta}).
\end{equation}
The far-field steering matrix $\mathbf{A}_c^{\text{FF}}(\boldsymbol{\theta})$ has the Vandermonde form with entries $[e^{j2\pi(m-1)d_c\sin\theta_k/\lambda}]_{m,k}$. A fundamental property of a Vandermonde matrix is that it has full column rank when all generating points (here, the spatial frequencies $\sin\theta_k$) are distinct. Therefore, $\text{rank}(\mathbf{A}_c^{\text{FF}}) = K$ when condition (i) is satisfied.

Since $\mathbf{C}_c$ is full-rank (from the proof of Lemma~\ref{lem:compressed_dof}) and $\mathbf{F}$ has full row rank, the composition $\mathbf{F}\mathbf{C}_c$ acts as a full-rank linear transformation on the first $M-2p$ dimensions. This preserves the column rank of $\mathbf{A}_c^{\text{FF}}$ as long as $K \leq M - 2p$. Thus, $\text{rank}(\mathbf{A}_c(\boldsymbol{\theta})) = K$ when both conditions (i) and (ii) hold, confirming that the bound $K_{\max,c}^{\text{FF}} = M - 2p - 1$ is tight.
\end{proof}

\begin{corollary}
For the baseline S-FAS implementation with $M = 32$ and $p = 3$, the compressed configuration can uniquely identify up to $K_{\max,c}^{\text{FF}} = 25$ far-field sources, representing a $19.2\%$ reduction compared to the theoretical limit $M - 1 = 31$ for an ideal coupling-free array.
\end{corollary}

To empirically corroborate Theorem~\ref{thm:compressed_identifiability} under realistic mutual coupling, we examine the noise subspace dimension in the compressed configuration with $M = 32$, $p = 3$, $d_c = 0.25\lambda$, $|c_1| \approx 0.17$, SNR = 20~dB, and $N = 1000$ snapshots. For each source number $K$, we construct the coupled-and-selected manifold, generate snapshots, form the sample covariance, and compute its eigenvalues.

Fig.~\ref{fig:enumeration_noise} shows that the simulated average noise subspace dimension perfectly matches the theoretical prediction $M_c - K$. The red vertical line marks the identifiability bound $K_{\max,c} = M_c - 1 = 25$, at which the noise subspace has dimension one, and the gray vertical line marks $K = M_c = 26$, where the noise subspace collapses to zero. This confirms that at least one DoF must be reserved for the noise baseline in the compressed configuration, so the practical identifiability limit coincides with the algebraic bound $K_{\max,c}^{\text{FF}} = M - 2p - 1$.

\begin{figure}[!t]
\centering
\includegraphics[width=0.52\textwidth]{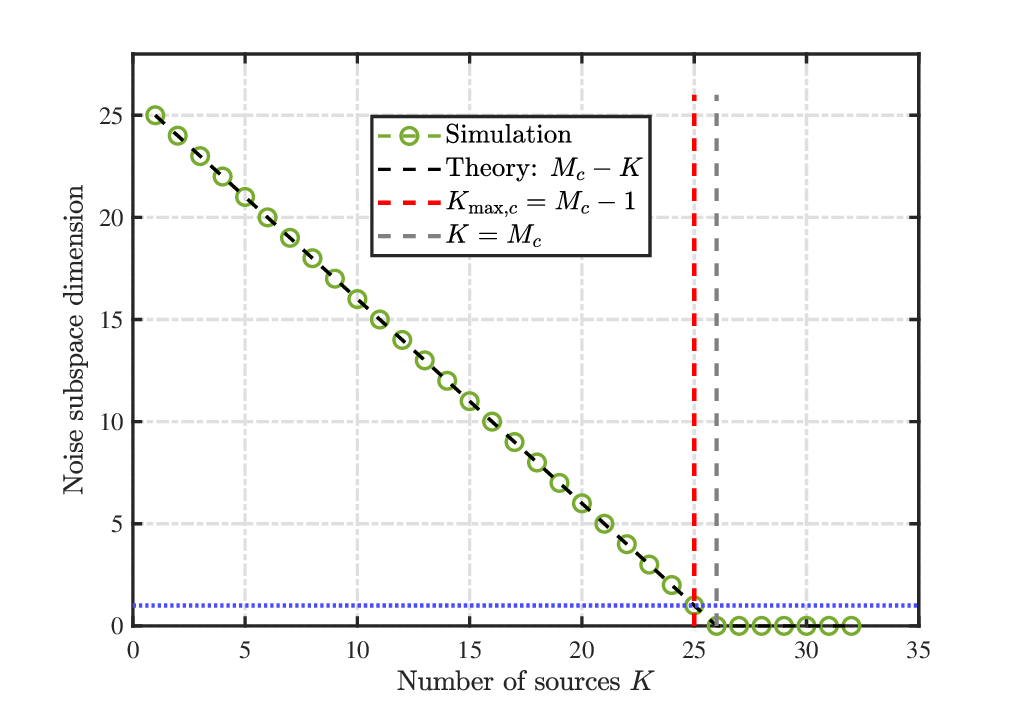}
\caption{Compressed configuration: Algebraic validation via $\dim(\mathbf{U}_n)$ versus number of Sources $K$ under realistic mutual coupling.}
\label{fig:enumeration_noise}
\end{figure}

Having confirmed the algebraic bound through noise subspace dimension measurements, we now provide an independent information-theoretic validation of the same identifiability limit, as summarized in the following remark and Fig.~\ref{fig:compressed_entropy_ratio}.

\begin{remark} 
\label{rem:compressed_info_theory}
The identifiability bound $K_{\max,c}^{\text{FF}} = M - 2p - 1$ can be independently verified through entropy analysis. Consider the mutual information between source parameters $\boldsymbol{\Theta} = \{\theta_1, \ldots, \theta_K\}$ and the compressed configuration observations $\mathbf{Y}_c$:
\begin{equation}
I(\boldsymbol{\Theta}; \mathbf{Y}_c) = H(\mathbf{Y}_c) - H(\mathbf{Y}_c|\boldsymbol{\Theta}).
\end{equation}
After central subarray selection and coupling compensation, the effective observation space has dimensionality $M_c = M - 2p$, yielding an observation entropy upper bound
\begin{equation}
H(\mathbf{Y}_c) \leq M_c \cdot \log(2\pi e \sigma_y^2) = (M - 2p) \cdot \log(2\pi e \sigma_y^2).
\end{equation}
 The covariance EVD partitions the entropy between signal and noise subspaces:
 \begin{align}
 H(\boldsymbol{\Lambda}_c) &= \sum_{i=1}^K \log(\lambda_i) + \sum_{i=K+1}^{M_c} \log(\sigma_n^2)\nonumber\\
 & = H_{\text{signal}}(K) + (M_c - K)\log(\sigma_n^2)
 \end{align}
 where the noise subspace entropy $(M_c - K)\log(\sigma_n^2)$ serves as the statistical reference baseline. When $K = M_c$, this entropy term vanishes, eliminating the noise reference required for parameter discrimination. Thus, reliable identifiability requires
 \begin{equation}
 M_c - K \geq 1 \quad \Longrightarrow \quad K \leq M_c - 1 = (M - 2p) - 1
 \end{equation}
 which independently confirms $K_{\max,c}^{\text{FF}} = M - 2p - 1$ from an information-theoretic perspective. This entropy constraint reflects the fundamental requirement that at least one DoF must be allocated to characterize noise statistics, without which signal components cannot be reliably distinguished from random fluctuations.
 \end{remark}

To empirically validate this information-theoretic constraint in the compressed configuration using the same noise entropy ratio as in Fig.~\ref{fig:validate_mutual_information}, we perform Monte Carlo simulations with moderate sub-wavelength spacing $d_c = 0.25\lambda$ (yielding $M_c = 26$ effective elements and $|c_1| \approx 0.17$), SNR = 20~dB, and $N = 1000$ snapshots. For each source number $K$, we construct the coupled-and-selected manifold, form the sample covariance matrix, and decompose the eigenvalue entropy into signal and noise contributions $H_{\text{signal}}(K)$ and $H_{\text{noise}}(K)$.

Fig.~\ref{fig:compressed_entropy_ratio} shows that the noise entropy ratio decays monotonically as the number of sources increases and approaches zero as $K$ reaches $M_c = 26$, with a sharp drop beyond the identifiability bound $K_{\max,c} = M_c - 1 = 25$. Once the noise subspace vanishes, no entropy can be allocated to the noise baseline, so the compressed observation space loses the statistical reference required to distinguish signal components from noise, confirming the entropy-based constraint underlying Theorem~\ref{thm:compressed_identifiability}.

 \begin{figure}[!t]
 \centering
 \includegraphics[width=0.52\textwidth]{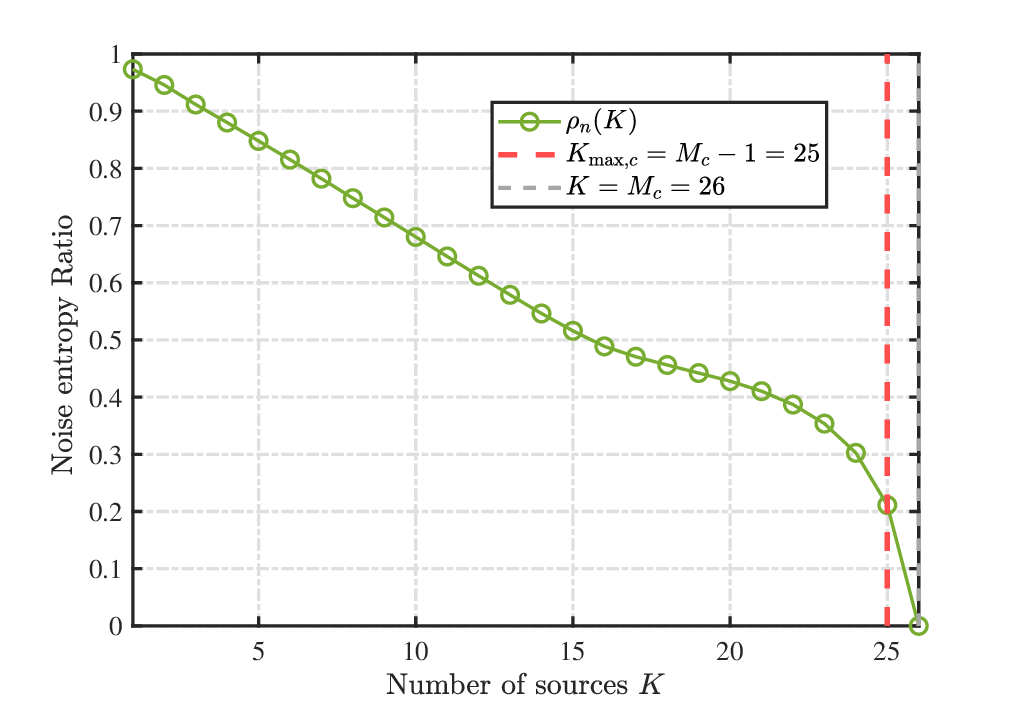}
 \caption{Compressed configuration: Information-theoretic validation via noise entropy ratio $\rho_n(K)$.}
 \label{fig:compressed_entropy_ratio}
 \end{figure}

\subsection{Grating-Lobe-Free Property}

A crucial advantage of the compressed configuration is its immunity to grating lobe ambiguities. Conventional arrays with half-wavelength spacing ($d = \lambda/2$) can suffer from grating lobes at endfire directions, and arrays with spacing $d > \lambda/2$ suffer from grating lobes over the visible region. The compressed configuration's sub-wavelength spacing eliminates this fundamental limitation:

\begin{proposition} 
\label{prop:grating_lobe_free}
For element spacing $d_c < \lambda/2$, the compressed array's spatial response is free from grating lobes over the complete visible region $\theta \in [-90^\circ, 90^\circ]$.
\end{proposition}

\begin{proof}
In the spatial frequency domain, the array response can be viewed as a periodic function with period $2\pi/d_c$ in the spatial frequency variable $k_x = (2\pi/\lambda)\sin\theta$. Grating lobes—spurious peaks in the array response pattern—occur when this periodic structure causes ambiguity, specifically when
\begin{equation}\label{eq:grating_condition}
k_x \cdot d_c = 2\pi n, \quad n \in \mathbb{Z} \setminus \{0\}.
\end{equation}

For sources in the visible region, the spatial frequency is bounded by $|k_x| = (2\pi/\lambda)|\sin\theta| \leq 2\pi/\lambda$ (with equality at endfire $\theta = \pm 90^\circ$). The grating lobe condition \eqref{eq:grating_condition} requires
\begin{equation}
|k_x \cdot d_c| = \left|\frac{2\pi}{\lambda}\sin\theta\right| \cdot d_c \geq 2\pi.
\end{equation}
For sub-wavelength spacing $d_c < \lambda/2$, we have
\begin{equation}
\max_{\theta} |k_x \cdot d_c| = \frac{2\pi}{\lambda} \cdot d_c < \frac{2\pi}{\lambda} \cdot \frac{\lambda}{2} = \pi < 2\pi.
\end{equation}
Thus, the grating lobe condition cannot be satisfied for any $n \geq 1$ within the visible region, confirming that the array response exhibits a unique main lobe for each source without ambiguous replicas.
\end{proof}

To empirically validate Proposition~\ref{prop:grating_lobe_free}, we evaluate the spatial frequency product $|k_x \cdot d_c| = (2\pi/\lambda)d_c |\sin\theta|$ across the full angular range $\theta \in [-90^\circ, 90^\circ]$ for the baseline compressed configuration with sub-wavelength spacing $d_c = 0.1\lambda$.

Fig.~\ref{fig:grating_lobe} shows that the spatial frequency product (blue curve) follows the magnitude of the sine function, reaching its maximum magnitude of $0.2\pi \approx 0.628~\text{rad}$ at endfire angles $\theta = \pm 90^\circ$ and vanishing at broadside $\theta = 0^\circ$. The red dashed line marks the grating lobe threshold $2\pi \approx 6.28~\text{rad}$, while the green shaded region highlights the grating-lobe-free safe zone. The blue curve remains strictly below the threshold across all angles, with a 20-fold safety margin. This substantial safety margin confirms that the compressed configuration is immune to grating lobe ambiguities, ensuring unambiguous DOA initialization in Stage~1 without spatial aliasing artifacts that could corrupt Stage~2 refinement.

\begin{figure}[!t]
\centering
\includegraphics[width=0.5\textwidth]{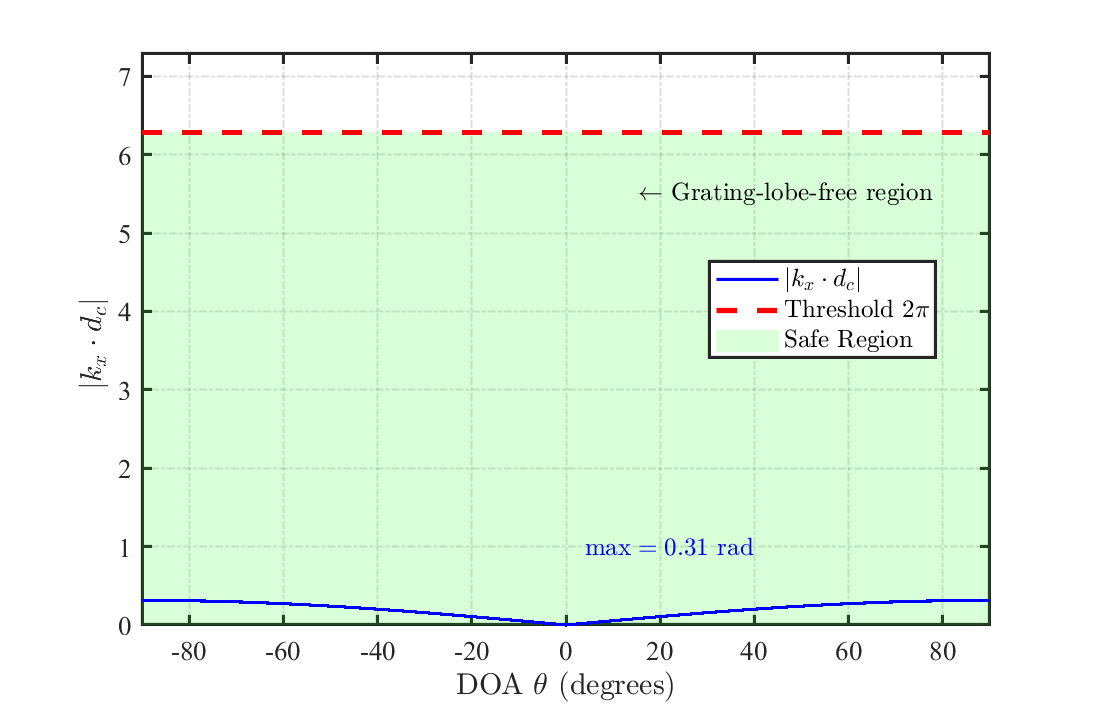}
\caption{Grating lobe immunity: Spatial frequency product versus DOA for compressed configuration ($d_c = 0.1\lambda$).}
\label{fig:grating_lobe}
\end{figure}

Proposition~\ref{prop:grating_lobe_free} guarantees that the compressed configuration can perform unambiguous DOA estimation over the full $180^\circ$ angular sector $\theta \in [-90^\circ, 90^\circ]$ without spatial aliasing. This property is critical for Stage~1 initialization in the two-stage S-FAS framework: even with coarse angular resolution, the compressed array provides reliable initial DOA estimates free from ambiguous peaks that could lead to catastrophic initialization errors in Stage~2. However, this grating-lobe-free operation comes at the cost of reduced angular resolution.

\section{Identifiability Analysis: Extended Configuration}

Section~III specializes the observation entropy framework under the \textit{entropy budget reduction} mechanism, establishing $K_{\max,c}^{\text{FF}} = M - 2p - 1$. We now apply the framework to the extended configuration under the \textit{full entropy budget} mechanism (Table~\ref{tab:framework_roadmap}), where $M_{\text{eff},e} = M$ and no entropy is sacrificed for coupling mitigation. The extended configuration operates with inter-element spacing $d_e = \alpha_e d_0 \approx \lambda/2$, achieving enhanced spatial resolution and negligible mutual coupling. A critical distinction from the compressed mode is that the extended configuration must handle \textit{mixed-field} localization: sources may be in near-field, Fresnel, or far-field regions, requiring joint angle-range estimation rather than DOA-only estimation. This increased parameter dimensionality fundamentally alters the identifiability analysis.

\subsection{Parameter Dimensionality}

The extended configuration with $M_{\text{eff},e} = M$ elements (no edge removal due to negligible coupling at $d_e \approx \lambda/2$) operates in two distinct scenarios depending on the source field regime:
\begin{itemize}
\item \textbf{Far-field}: Each source has one parameter $(\theta_k)$. Manifold is Vandermonde $\mathbf{A}_e^{\text{FF}}(\boldsymbol{\theta})$.
\item \textbf{Mixed-field}: Each source has two parameters $(\theta_k, r_k)$ coupled via the ESG model.
\end{itemize}
For joint estimation, $2K$ geometric unknowns (angles and ranges) must be recovered from the $M(M+1)/2$ unique entries in the Hermitian covariance matrix. This parameter counting might naively suggest that mixed-field estimation will support fewer sources than far-field DOA-only estimation due to the doubled parameter dimensionality. We now show that this naive expectation is, in fact, incorrect.

\subsection{Mixed-Field Identifiability}

When sources exist in the near-field or Fresnel regions, the ESG model must be employed, requiring joint estimation of both DOA and range for each source. Although each source now involves two parameters $(\theta_k, r_k)$ rather than one, the identifiability bound remains determined by the signal subspace dimension:

\begin{theorem} 
\label{thm:extended_mixed_field}
For the extended configuration operating in mixed-field mode with the ESG steering model $\mathbf{A}_e^{\text{ESG}}(\boldsymbol{\theta}, \mathbf{r})$, the maximum number of jointly identifiable source parameter pairs $\{(\theta_k, r_k)\}_{k=1}^K$ is
\begin{equation}\label{eq:kmax_extended_mf}
K_{\max,e}^{\text{MF}} = M - 1
\end{equation}
provided that the source parameters satisfy separability conditions ensuring the array manifold has full column rank. Notably, this bound is \textit{identical} to the far-field case despite each source now involving two parameters instead of one.
\end{theorem}

\begin{proof}
The key insight is that the identifiability constraint is determined by the \textit{array manifold dimension}, not the \textit{parameter count}. Although each source now has two parameters $(\theta_k, r_k)$, the array manifold $\mathbf{A}_e^{\text{ESG}}(\boldsymbol{\theta}, \mathbf{r}) \in \mathbb{C}^{M \times K}$ remains a matrix with $K$ columns---one column per source, regardless of how many parameters define each column.

The received signal covariance matrix is
\begin{equation}
\mathbf{R}_e = \mathbf{A}_e^{\text{ESG}}(\boldsymbol{\theta}, \mathbf{r})\mathbf{R}_s\mathbf{A}_e^{\text{ESG}}(\boldsymbol{\theta}, \mathbf{r})^H + \sigma_n^2\mathbf{I}_M.
\end{equation}
The EVD partitions the observation space into:
\begin{itemize}
\item \textbf{Signal subspace} $\mathbf{U}_s \in \mathbb{C}^{M \times K}$: spanned by $K$ dominant eigenvectors, with $\text{span}(\mathbf{U}_s) = \text{span}(\mathbf{A}_e^{\text{ESG}})$.
\item \textbf{Noise subspace} $\mathbf{U}_n \in \mathbb{C}^{M \times (M-K)}$: spanned by $M-K$ remaining eigenvectors, orthogonal to $\mathbf{A}_e^{\text{ESG}}$.
\end{itemize}

Subspace methods (MUSIC, ESPRIT variants for near-field) exploit the orthogonality condition
\begin{equation}
\mathbf{a}_e^{\text{ESG}}(\theta_k, r_k)^H \mathbf{U}_n = \mathbf{0}, \quad k=1,\ldots,K
\end{equation}
to search over the two-dimensional parameter space $(\theta, r)$. This requires the noise subspace to exist, requiring
\begin{equation}
M - K \geq 1 \quad \Longrightarrow \quad K \leq M - 1.
\end{equation}

Crucially, this constraint is \textit{independent} of the number of parameters per source. Whether estimating $K$ angles (far-field) or $K$ angle-range pairs (near-field), the manifold dimension remains $K$, and the noise subspace dimensionality requirement yields the same bound $K \leq M-1$ in both cases.

This bound is tight when the Jacobian matrix of the steering vector with respect to $(\boldsymbol{\theta}, \mathbf{r})$ has full rank $2K$, which holds when sources are sufficiently separated in both angle and range to avoid parameter ambiguities.

\textit{Information-theoretic perspective:} From an information-theoretic viewpoint, the identifiability constraint can be understood through the mutual information between the unknown parameters $\boldsymbol{\Theta} = \{(\theta_k, r_k)\}_{k=1}^K$ and the observations $\mathbf{Y} \in \mathbb{C}^{M \times N}$. The maximum achievable mutual information is bounded by the observation entropy:
\begin{equation}
I(\boldsymbol{\Theta}; \mathbf{Y}) \leq H(\mathbf{Y}) \leq MN\log(2\pi e \sigma_{\max}^2)
\end{equation}
where $\sigma_{\max}^2$ is the maximum eigenvalue of $\mathbf{R}_e$. Crucially, this upper bound depends on the observation dimension $M$, not the parameter count $2K$. The signal lies in a $K$-dimensional subspace of the $M$-dimensional observation space. When $K = M$, all DoFs are consumed by the signal, leaving no reference dimension to distinguish signal from noise. The Fisher information matrix $\mathbf{J}(\boldsymbol{\Theta}) \in \mathbb{R}^{2K \times 2K}$ for the $2K$ parameters can only be full rank (ensuring local identifiability) when the noise subspace provides orthogonality constraints. This requires:
\begin{equation}
\text{rank}(\mathbf{U}_n) = M - K \geq 1 \quad \Longrightarrow \quad K \leq M - 1.
\end{equation}
The bound is independent of whether each source contributes 1 parameter (far-field $\theta$) or 2 parameters (near-field $\theta, r$)—the constraint arises from the \textit{observational DoFs} $M$, not the \textit{parametric DoFs} $2K$.
\end{proof}

\begin{remark} 
\label{rem:extended_info_theory}
The identifiability bound $K_{\max,e} = M - 1$ can be independently verified through entropy decomposition. For the extended configuration with observation space dimensionality $M$, the mutual information between source parameters $\boldsymbol{\Theta}$ (either $K$ angles for far-field or $K$ angle-range pairs for mixed-field) and observations $\mathbf{Y}_e$ satisfies
\begin{equation}
I(\boldsymbol{\Theta}; \mathbf{Y}_e) = H(\mathbf{Y}_e) - H(\mathbf{Y}_e|\boldsymbol{\Theta}) \leq H(\mathbf{Y}_e) \leq M \cdot \log(2\pi e \sigma_y^2).
\end{equation}
The EVD partitions the observation entropy between signal and noise subspaces:
\begin{align}
H(\boldsymbol{\Lambda}_e) &= \sum_{i=1}^K \log(\lambda_i) + \sum_{i=K+1}^{M} \log(\sigma_n^2) \nonumber\\
& = H_{\text{signal}}(K) + (M - K)\log(\sigma_n^2)
\end{align}
where the noise subspace entropy $(M - K)\log(\sigma_n^2)$ provides the statistical reference baseline. When $K = M$, this term vanishes, eliminating the noise reference required for parameter discrimination. Thus, reliable identifiability requires
\begin{equation}
M - K \geq 1 \quad \Longrightarrow \quad K \leq M - 1
\end{equation}
regardless of whether each source contributes 1 parameter (far-field) or 2 parameters (mixed-field) to $\boldsymbol{\Theta}$.
\end{remark}

\subsection{Far-Field Identifiability}

When all sources are sufficiently distant to satisfy the far-field condition $r_k \gg 2D_e^2/\lambda$, the range parameter becomes irrelevant and estimation reduces to DOA-only, recovering the classical ULA identifiability bound:

\begin{corollary} 
\label{cor:extended_far_field}
When all sources are in the far-field region and the array manifold is $\mathbf{A}_e^{\text{FF}}(\boldsymbol{\theta})$, the maximum number of identifiable DOAs is
\begin{equation}\label{eq:kmax_extended_ff}
K_{\max,e}^{\text{FF}} = M - 1.
\end{equation}
\end{corollary}

\begin{proof}
This follows directly from Proposition~\ref{prop:fundamental_constraint} with $M_{\text{eff}} = M$ for the extended configuration (no edge element removal). The far-field steering matrix has Vandermonde structure with full column rank for distinct DOAs, satisfying the rank condition. The noise subspace dimension is $M - K$, requiring $K \leq M - 1$ for subspace-based estimation.
\end{proof}

Corollary~\ref{cor:extended_far_field} shows that far-field-only processing in the extended configuration achieves the classical ULA limit $K_{\max,e}^{\text{FF}} = M - 1$. This result underscores the importance of field regime classification in practice: if sources can be reliably identified as far-field through auxiliary information or range pre-filtering, the system should operate in far-field-only mode to maximize capacity. However, the S-FAS framework's key innovation is that it does not \textit{require} such \textit{a priori} classification—the ESG model handles all field regimes uniformly with identical theoretical capacity, though at increased algorithmic complexity.

\subsection{Simulation Validation}

To empirically validate Theorem~\ref{thm:extended_mixed_field} and Remark~\ref{rem:extended_info_theory}, we perform Monte Carlo simulations comparing far-field and mixed-field scenarios under identical conditions: $M = 32$ elements, $d_e = 0.5\lambda$, SNR = 20~dB, and $N = 1000$ snapshots. For each source number $K$, we construct the corresponding steering matrix (Vandermonde for far-field with 1 parameter $\theta$ per source, ESG for mixed-field with 2 parameters $(\theta, r)$ per source), generate observations, form the sample covariance, and compute its eigenvalues.

\subsubsection{Algebraic Validation}

We use the minimum description length (MDL) criterion \cite{MDL} to estimate $K$ from eigenvalues, then compute $\dim(\mathbf{U}_n) = M - K_{\text{est}}$ and average over Monte Carlo trials. Fig.~\ref{fig:extended_algebraic} shows that both far-field and mixed-field scenarios track the theoretical prediction $M - K$ nearly perfectly, demonstrating that the noise subspace dimension depends on the manifold column count $K$, not the parameter count per source. At the identifiability boundary $K_{\max,e} = M - 1 = 31$, the noise subspace has dimension 1, providing the minimal orthogonality reference required by subspace methods. At $K = M = 32$, the noise subspace vanishes completely, eliminating this reference. Near the theoretical limit $K \approx M-1$, the mixed-field curve appears slightly above the far-field curve due to finite-snapshot MDL estimation bias: the ESG near-field manifold yields a more uneven eigenvalue spectrum and a more ill-conditioned covariance matrix, so MDL tends to mildly underestimate $K$ in the mixed-field case at very high loading, resulting in a marginally larger estimated noise subspace dimension. This small discrepancy reflects the increased algorithmic difficulty of mixed-field estimation rather than any change in the fundamental identifiability bound, which remains $K_{\max,e} = M - 1$ for both scenarios.

\begin{figure}[!t]
\centering
\includegraphics[width=0.52\textwidth]{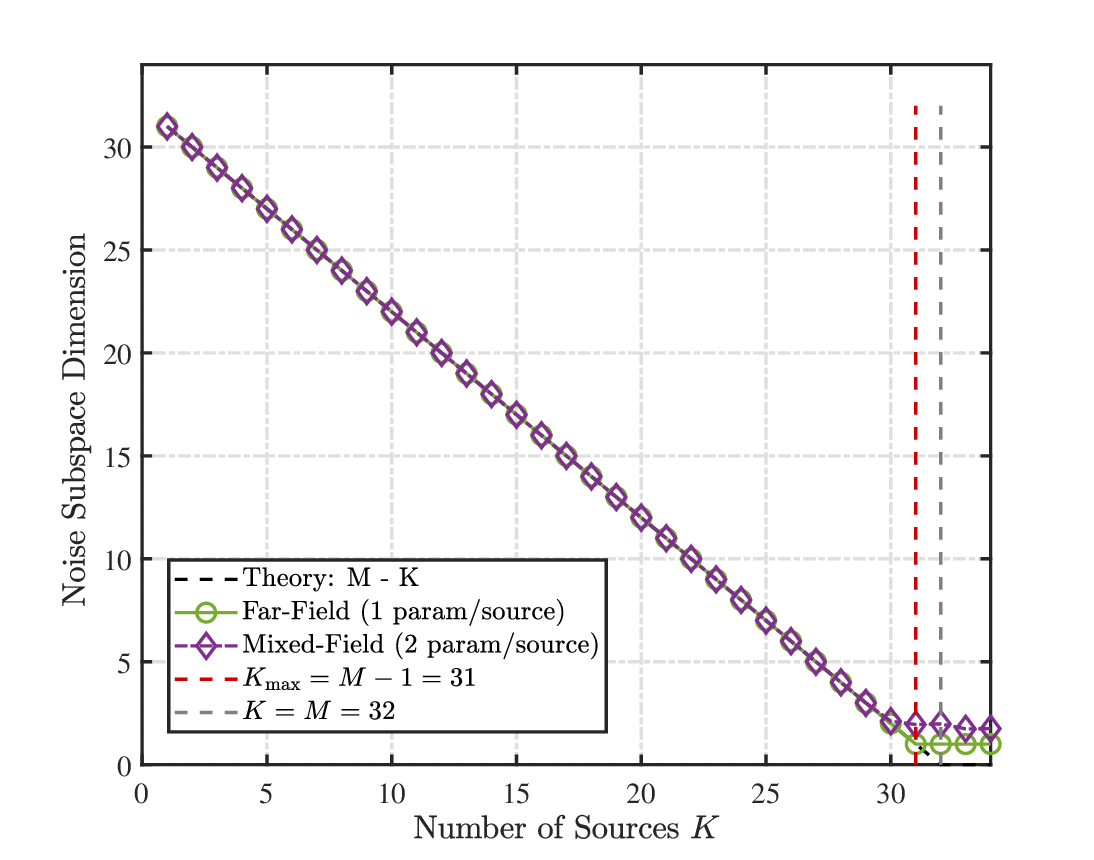}
\caption{Extended configuration: Algebraic validation via $\dim(\mathbf{U}_n)$ for far-field and mixed-field scenarios.}
\label{fig:extended_algebraic}
\end{figure}

\subsubsection{Information-Theoretic Validation}

To validate the information-theoretic constraint underlying Remark~\ref{rem:extended_info_theory}, we compute the noise entropy ratio for both far-field and mixed-field scenarios. For each $K$, we decompose the eigenvalue entropy into signal and noise contributions, then calculate the fraction allocated to the noise baseline.

Fig.~\ref{fig:extended_entropy_ratio} shows that this ratio decays monotonically as $K$ increases and approaches zero as $K \to M$, with a sharp drop beyond $K_{\max,e} = 31$ for both scenarios. Once the noise subspace vanishes at $K = M$, no entropy can be allocated to the noise baseline, eliminating the statistical reference required to distinguish signal components from noise. Compared with the almost linear decay in the far-field case (blue curve), the mixed-field curve (red) exhibits a pronounced convex shape---$\rho_{\text{n}}(K)$ stays close to one over a wide range of $K$ and then drops abruptly only when $K$ approaches $M$. This behavior reflects the much more uneven eigenvalue spectrum of the ESG manifold: a few dominant eigenmodes capture most of the signal energy, while many remaining signal eigenvalues are comparable to the noise floor, so the incremental contribution of additional sources to the signal entropy $H_{\text{signal}}(K)$ is smaller than in the far-field case. As a result, the noise entropy $H_{\text{noise}}(K)$ continues to dominate the total entropy for moderate $K$, keeping $\rho_{\text{n}}(K)$ high and nearly flat. Only when $K$ approaches $M$ do the last few signal modes consume the remaining DoFs and force a rapid collapse of the noise entropy ratio. This convex decay pattern therefore reveals that mixed-field estimation uses the available observational DoFs less uniformly and is algorithmically more challenging than far-field estimation, even though both scenarios share the same information-theoretic identifiability limit $K_{\max,e} = M - 1$.

\begin{figure}[!t]
\centering
\includegraphics[width=0.52\textwidth]{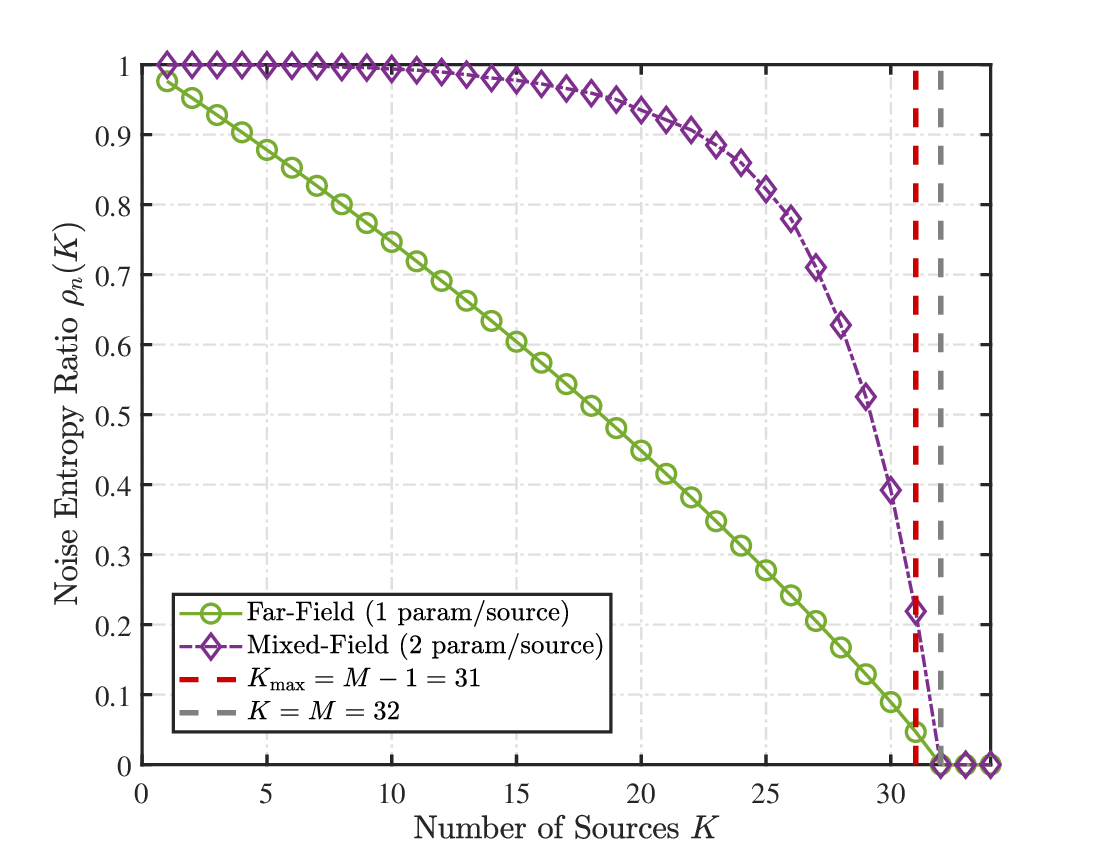}
\caption{Extended configuration: Information-theoretic validation via noise entropy ratio $\rho_n(K)$ for far-field and mixed-field scenarios.}
\label{fig:extended_entropy_ratio}
\end{figure}

\section{S-FAS Capacity: Sequential vs Joint Processing}
\label{sec:sfas_capacity}

Sections~III and~IV specialize Pillars~1 and~2 of the observation entropy framework to individual configurations, establishing $K_{\max,c} = M-2p-1$ (entropy budget reduction) and $K_{\max,e} = M-1$ (full entropy budget). We now develop \textit{Pillar~3}—the entropy hierarchy—which governs how identifiability changes when both configurations are used together (Table~\ref{tab:framework_roadmap}). Two architectures are possible: \textit{sequential processing}, where the compressed entropy budget limits the information available to Stage~2 via the data processing inequality, and \textit{joint processing}, where spatial stacking expands the entropy budget to $H(\mathbf{Y}_{\text{joint}}) \propto (M_c + M_e)$. This section derives the theoretical capacity limits of both approaches, revealing a fundamental trade-off: sequential processing suffers from an entropy bottleneck, while joint processing achieves substantially higher capacity through entropy expansion.

\subsection{Sequential Two-Stage Architecture}
\label{sec:two_stage_capacity}

In the practical sequential S-FAS implementation, Stage~1 operates on the compressed observations $\mathbf{Y}_c$ to produce DOA estimates $\hat{\boldsymbol{\theta}}_c$, which are then passed as initialization to Stage~2 for refinement using the extended observations $\mathbf{Y}_e$. Since Stage~2 cannot enumerate sources that were missed by Stage~1, the end-to-end capacity is limited by the compressed-stage bound:
\begin{corollary} \label{cor:sequential_capacity}
For the two-stage S-FAS architecture with compressed and extended configurations characterized above, the maximum number of identifiable sources is
\begin{align}\label{eq:kmax_sfas_sequential}
K_{\text{S-FAS}}^{\text{sequential}} &= \min(K_{\max,c}^{\text{FF}}, K_{\max,e}^{\text{MF}})\nonumber\\
 &= \min(M-2p-1, M-1) = M-2p-1
\end{align}
for mixed-field scenarios where joint angle-range estimation is required, and
\begin{align}
K_{\text{S-FAS}}^{\text{FF-only}} &= \min(K_{\max,c}^{\text{FF}}, K_{\max,e}^{\text{FF}}) \nonumber\\
&= \min(M-2p-1, M-1) = M-2p-1
\end{align}
for far-field-only refinement. In both cases, the end-to-end sequential capacity is therefore limited by the compressed configuration's identifiability bound $K_{\max,c} = M-2p-1$.
\end{corollary}

To validate the sequential capacity bottleneck, we perform Monte Carlo simulations ($M=32$, $p=3$, $M_c=26$, $M_e=32$ and SNR = 30~dB) comparing three processing modes: (i) compressed-only Stage~1, (ii) extended-only single-stage, and (iii) sequential two-stage (Stage~1 compressed MDL enumeration → Stage~2 extended refinement). For each $K$, we measure the average noise subspace dimension $\dim(\mathbf{U}_n)$ after processing.

Fig.~\ref{fig:sequential_algebraic} shows that the compressed and extended configurations track their theoretical predictions $M_c - K$ and $M_e - K$. Critically, the sequential architecture inherits the compressed configuration's noise subspace dimension rather than exploiting the extended capacity; the magenta curve follows the blue curve, demonstrating that $\dim(\mathbf{U}_n)$ is limited by Stage~1's MDL enumeration. The shaded region highlights 6 sources ($K \in [26, 31]$) that the extended configuration can theoretically resolve but remain inaccessible due to the Stage~1 initialization bottleneck.

\begin{figure}[!t]
\centering
\includegraphics[width=0.48\textwidth]{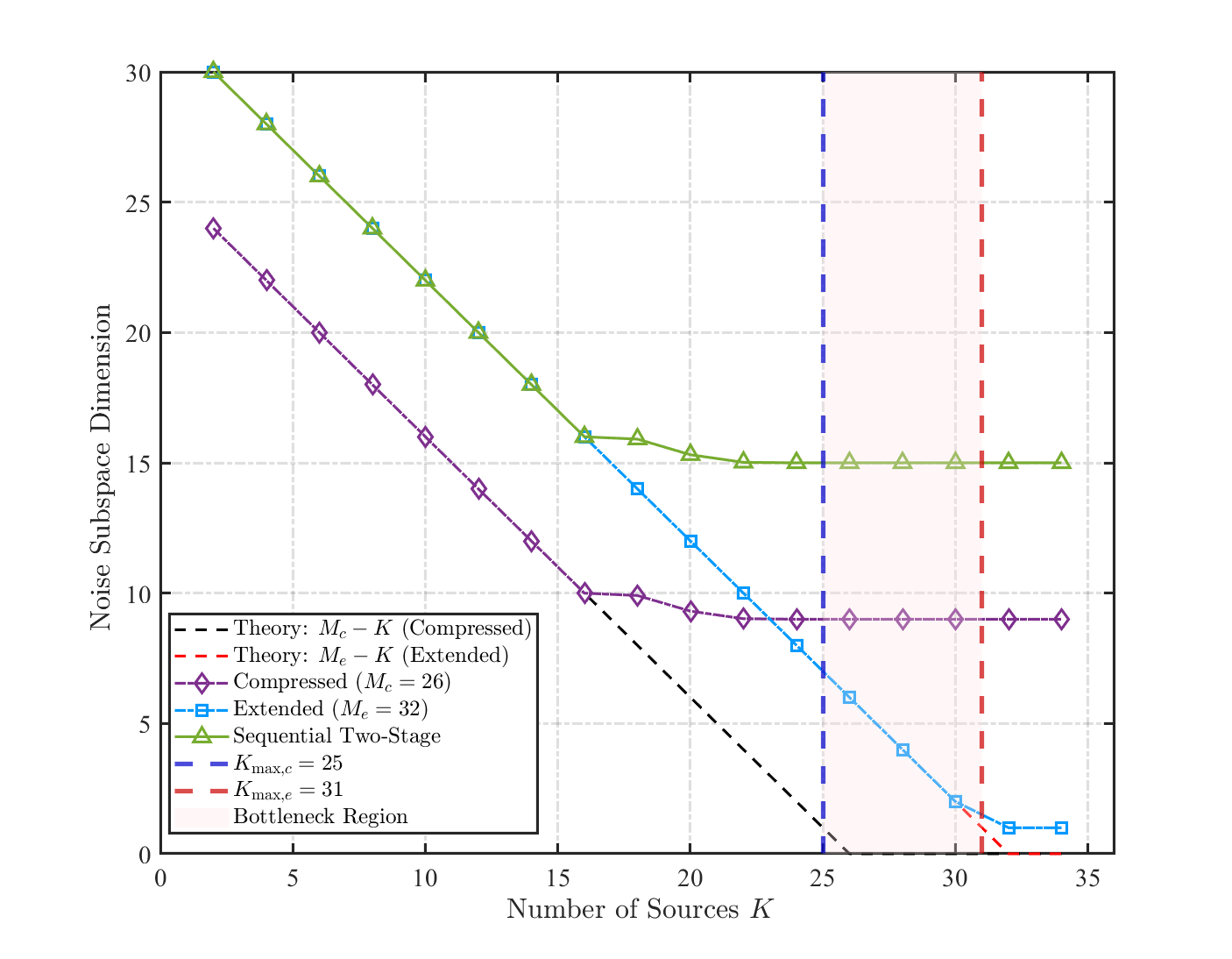}
\caption{Sequential architecture: Algebraic validation showing Stage~1 bottleneck via $\dim(\mathbf{U}_n)$ for compressed, extended, and sequential configurations.}
\label{fig:sequential_algebraic}
\end{figure}

  For moderate source counts ($K < 20$), the compressed and extended curves closely follow the DoF-based predictions $M_c - K$ and $M_e - K$. As $K$ approaches $K_{\max,c} = 25$, however, the MDL enumeration on the compressed array saturates and $\dim(\mathbf{U}_n)$ for the compressed configuration levels off at approximately nine dimensions. The sequential two-stage curve remains roughly $M_e - M_c = 6$ dimensions above the compressed curve, i.e., $\dim(\mathbf{U}_n^{\text{seq}}) \approx \dim(\mathbf{U}_n^{\text{c}}) + (M_e - M_c)$, indicating that the additional six extended elements are effectively converted into unused noise-subspace DoFs once Stage~1 has saturated. In contrast, the extended configuration operating alone continues to drive $\dim(\mathbf{U}_n)$ down in accordance with $M_e - K$ up to $K_{\max,e} = 31$, underscoring that the loss of algebraic capacity is due entirely to the sequential initialization constraint rather than any limitation of the extended manifold itself.

Having established the algebraic capacity bound through the minimum operator, we now provide an information-theoretic explanation for why sequential processing cannot exceed the compressed-stage limit, even when the extended configuration offers a larger aperture.

\begin{remark} 
Corollary~\ref{cor:sequential_capacity} formalizes the sequential bottleneck: the end-to-end capacity collapses to $K_{\text{S-FAS}}^{\text{sequential}} = M-2p-1 = 25$ sources even though the extended configuration alone supports $K_{\max,e}^{\text{FF}} = M-1 = 31$ sources. From an information-theoretic viewpoint, the data-processing inequality
\begin{equation}
I(\boldsymbol{\theta}; \hat{\boldsymbol{\theta}}_c) \leq I(\boldsymbol{\theta}; \mathbf{Y}_c) \leq H(\mathbf{Y}_c) \leq M_c N \log(2\pi e \sigma_y^2),
\end{equation}
shows that the mutual information available to Stage~2 is fundamentally bounded by $M_c = M-2p$, irrespective of the extended configuration's larger aperture. Any sequential architecture that discards $\mathbf{Y}_c$ after producing $\hat{\boldsymbol{\theta}}_c$ necessarily wastes the extended array's observational DoFs and cannot exceed $K_{\max} = M-2p-1$.
\end{remark}

To validate the information-theoretic constraint, we compute the noise entropy ratio for all three processing modes under identical simulation conditions ($M=32$, $p=3$, and SNR = 40~dB).

Fig.~\ref{fig:sequential_info_theory} shows that the sequential curve tracks the compressed configuration rather than extended, confirming that the data-processing inequality $I(\boldsymbol{\theta}; \hat{\boldsymbol{\theta}}_c) \leq H(\mathbf{Y}_c)$ limits Stage~2's information to the compressed observation space. Beyond $K_{\max,c} = 25$, the noise entropy ratio collapses for compressed and sequential architectures while the extended configuration maintains higher entropy up to $K_{\max,e} = 31$. The shaded region ($K \in [26, 31]$) highlights the information bottleneck: these 6 sources are informationally inaccessible to sequential processing because the Stage~1 compressed observations contain insufficient entropy to discriminate them, even though the Stage~2 extended manifold has the geometric capacity.

\begin{figure}[!t]
\centering
\includegraphics[width=0.48\textwidth]{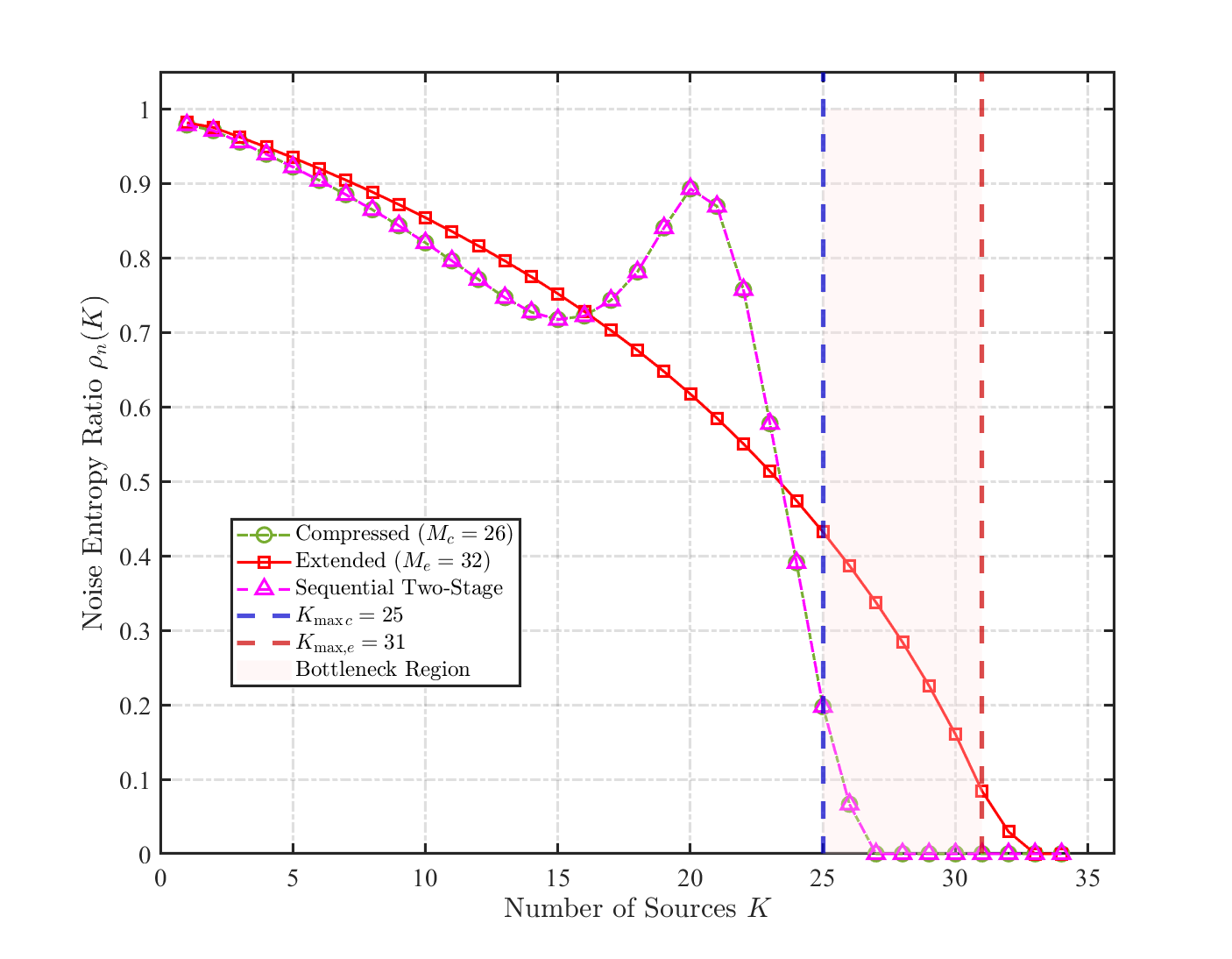}
\caption{Sequential architecture: Information-theoretic validation showing Stage~1 bottleneck via noise entropy ratio $\rho_n(K)$.}
\label{fig:sequential_info_theory}
\end{figure}

From Fig.~\ref{fig:sequential_info_theory} it can be observed that the compressed and sequential curves exhibit a mild local increase in $\rho_{\text{n}}(K)$ around $K \approx 20$. This behavior arises because, for the compressed array with $M_c = 26$ and inter-element spacing $d_c = 0.25\lambda$, the source spacing in this regime (roughly $6^\circ$ at $K \approx 20$) is already close to the array's resolvability limit. The sample covariance eigenvalue spectrum therefore becomes highly ill-conditioned, with weak signal modes leaking into the nominal noise subspace. When $\rho_{\text{n}}(K)$ is computed by partitioning the ordered eigenvalues according to the true source number $K$, small reallocations of these borderline eigenvalues between the signal and noise sets lead to non-monotonic fluctuations in the entropy ratio. Importantly, this localized variation occurs well below the theoretical compressed-stage bound and does not affect the dominant trend: at $K = K_{\max,c} = 25$ the noise entropy ratio for the compressed and sequential configurations collapses while the extended configuration still preserves a non-negligible noise reference, clearly illustrating that the end-to-end sequential capacity is bottlenecked by the compressed stage.

\subsection{Joint Configuration Processing}
\label{sec:joint_capacity}

Sequential processing, characterized in Section~\ref{sec:two_stage_capacity}, is thus limited to $K_{\text{S-FAS}}^{\text{sequential}} = K_{\max,c}^{\text{FF}} = M-2p-1$ sources by the compressed-stage initialization. In contrast, joint processing of both configurations retains the full observation vectors from the compressed and extended arrays and provides effective DoFs $M_c + M_e = (M-2p) + M = 2M-2p$, potentially resolving up to $2M-2p-1$ sources.

\subsection{Joint Signal Model}

Let $\mathbf{Y}_c \in \mathbb{C}^{M_c \times N}$ and $\mathbf{Y}_e \in \mathbb{C}^{M_e \times N}$ denote observations from both configurations. For far-field DOA estimation where sources share common angles $\{\theta_k\}_{k=1}^K$ across configurations, we construct an augmented observation vector by spatially stacking the measurements:
\begin{equation}\label{eq:joint_observation}
\mathbf{y}_{\text{joint}}(t) = \begin{bmatrix} \mathbf{y}_c(t) \\ \mathbf{y}_e(t) \end{bmatrix} \in \mathbb{C}^{(M_c + M_e)}
\end{equation}
with augmented array manifold
\begin{equation}\label{eq:joint_manifold}
\mathbf{A}_{\text{joint}}(\boldsymbol{\theta}) = \begin{bmatrix} \mathbf{A}_c(\boldsymbol{\theta}) \\ \mathbf{A}_e(\boldsymbol{\theta}) \end{bmatrix} \in \mathbb{C}^{(M_c + M_e) \times K}
\end{equation}
where $\mathbf{A}_c(\boldsymbol{\theta}) \in \mathbb{C}^{M_c \times K}$ is the compressed configuration manifold from \eqref{eq:effective_compressed_manifold} and $\mathbf{A}_e(\boldsymbol{\theta}) \in \mathbb{C}^{M_e \times K}$ is the extended configuration manifold from \eqref{eq:ff_steering_extended}. The joint signal model becomes
\begin{equation}\label{eq:joint_signal_model}
\mathbf{y}_{\text{joint}}(t) = \mathbf{A}_{\text{joint}}(\boldsymbol{\theta})\mathbf{s}(t) + \mathbf{n}_{\text{joint}}(t)
\end{equation}
where $\mathbf{n}_{\text{joint}}(t) = [\mathbf{n}_c^T(t), \mathbf{n}_e^T(t)]^T$ is the stacked noise vector.

\subsection{Identifiability Analysis}

\begin{theorem}[DoF-based Joint Identifiability Bound]
\label{thm:joint_dof_identifiability}
When both S-FAS configurations are jointly processed for far-field DOA-only estimation through spatial stacking, the maximum number of identifiable sources is
\begin{equation}\label{eq:kmax_joint}
K_{\max}^{\text{joint}} = (M_c + M_e) - 1 = 2M - 2p - 1
\end{equation}
where $M_c = M - 2p$ and $M_e = M$ are the effective dimensions of the compressed and extended configurations, respectively.
\end{theorem}

\begin{proof}
The augmented observation vector \eqref{eq:joint_observation} combines measurements from both configurations, yielding an effective spatial observation space of dimension $M_c + M_e = (M - 2p) + M = 2M - 2p$. The augmented covariance matrix is
\begin{equation}
\mathbf{R}_{\text{joint}} = E\left[\mathbf{y}_{\text{joint}}(t)\mathbf{y}_{\text{joint}}^H(t)\right] = \begin{bmatrix} \mathbf{R}_{cc} & \mathbf{R}_{ce} \\ \mathbf{R}_{ec} & \mathbf{R}_{ee} \end{bmatrix} \in \mathbb{C}^{(M_c+M_e) \times (M_c+M_e)}
\end{equation}
where $\mathbf{R}_{cc}$ and $\mathbf{R}_{ee}$ are the individual configuration covariances, while $\mathbf{R}_{ce}$ and $\mathbf{R}_{ec}$ capture cross-configuration correlations arising from common source parameters. Since sources share DOAs across configurations, the joint manifold \eqref{eq:joint_manifold} has full column rank $K$ when all DOAs are distinct (Vandermonde property preserved in stacking). Applying Proposition~\ref{prop:fundamental_constraint} with $M_{\text{eff}} = M_c + M_e$, the maximum number of identifiable sources is $(M_c + M_e) - 1 = 2M - 2p - 1$.

\textit{Information-theoretic perspective:} The capacity gain can be understood through mutual information and entropy bounds. The joint observation provides mutual information
\begin{equation}
I(\boldsymbol{\theta}; \mathbf{Y}_{\text{joint}}) = I(\boldsymbol{\theta}; \mathbf{Y}_c, \mathbf{Y}_e) = I(\boldsymbol{\theta}; \mathbf{Y}_c) + I(\boldsymbol{\theta}; \mathbf{Y}_e | \mathbf{Y}_c)
\end{equation}
by the chain rule. When the two observation sets are conditionally independent given the source parameters, $I(\boldsymbol{\theta}; \mathbf{Y}_e | \mathbf{Y}_c) = I(\boldsymbol{\theta}; \mathbf{Y}_e)$, so the joint mutual information equals the sum of individual contributions. Critically, the observation entropy bound scales with the joint dimension:
\begin{equation}
H(\mathbf{Y}_{\text{joint}}) \leq (M_c + M_e) N \log(2\pi e \sigma_{\max}^2)
\end{equation}
nearly doubling the single-configuration bounds $H(\mathbf{Y}_c) \leq M_c N\log(\cdot)$ and $H(\mathbf{Y}_e) \leq M_e N\log(\cdot)$. The noise subspace constraint requires $\text{rank}(\mathbf{U}_n^{\text{joint}}) = (M_c + M_e) - K \geq 1$, yielding $K \leq (M_c + M_e) - 1 = 2M - 2p - 1$. This demonstrates that joint processing exploits configuration diversity to expand the observational DoFs, fundamentally increasing capacity beyond what either configuration can achieve individually or sequentially.
\end{proof}

\begin{corollary}
The joint processing provides a capacity gain of
\begin{equation}
\frac{K_{\max}^{\text{joint}}}{K_{\max,e}^{\text{FF}}} = \frac{2M-2p-1}{M-1} = \frac{57}{31} \approx 1.84
\end{equation}
over the extended configuration alone, representing an $84\%$ capacity increase. For $M=32$ and $p=3$, this corresponds to a $128\%$ increase over the compressed configuration ($25 \to 57$) and an $84\%$ increase over the extended configuration ($31 \to 57$), nearly doubling the sequential bottleneck capacity of 25 sources.
\end{corollary}

However, the DoF-based bound in Theorem~\ref{thm:joint_dof_identifiability} is optimistic and relies on idealized assumptions about manifold rank and source enumeration. In practice, additional effects reduce the achievable joint capacity, as summarized in the following remark.

\begin{remark}[Theory-Practice Gap]\label{rem:theory_practice_gap}
The theoretical bound $K_{\max}^{\text{joint}} = 2M-2p-1$ assumes: (i) perfect manifold rank, i.e., $\text{rank}(\mathbf{A}_{\text{joint}}) = K$ for all $K \leq (M_c+M_e)-1$; and (ii) asymptotically optimal source enumeration. In practice, two factors limit achievable capacity:
\begin{enumerate}
\item \textbf{Manifold conditioning:} The compressed configuration manifold $\mathbf{A}_c$ with $d_c = 0.25\lambda$ exhibits poor conditioning for large $K$, with effective rank saturating around $16$--$20$ due to near-linear dependencies among steering vectors. Since $\text{rank}([\mathbf{A}_c; \mathbf{A}_e]) \leq \text{rank}(\mathbf{A}_c) + \text{rank}(\mathbf{A}_e)$, the joint manifold's effective rank is bounded by approximately $20 + 31 = 51$, below the theoretical $(M_c+M_e)-1 = 57$.
\item \textbf{MDL boundary behavior:} The MDL criterion becomes unstable when $K$ approaches $M_{\text{eff}}-1$, as the noise subspace dimension $(M_{\text{eff}}-K)$ shrinks to 1--2, providing insufficient statistical basis for distinguishing signal from noise eigenvalues.
\end{enumerate}
Monte Carlo simulations confirm practical joint capacity of $K \approx 34$--$36$, representing a $10$--$16\%$ gain over extended-only processing---still substantial, though below the theoretical $84\%$.
\end{remark}

\subsection{Validation: Algebraic and Information-Theoretic Perspectives}

We validate the joint processing capacity bounds from both algebraic (noise subspace dimension) and information-theoretic (normalized noise entropy) perspectives using comprehensive Monte Carlo simulations.

Fig.~\ref{fig:joint_algebraic} employs the noise subspace dimension $\dim(\mathbf{U}_n) = M_{\text{eff}} - K$ as the key metric. The simulation setting are SNR = 40~dB and $N = 8000$ snapshots to ensure statistical reliability. The simulated curves track theoretical bounds closely across all three configurations: compressed ($M_c = 26$, $K_{\max,c} = 25$), extended ($M_e = 32$, $K_{\max,e} = 31$), and joint (green, $M_{\text{joint}} = 58$, $K_{\max,joint} = 57$). \textbf{Critically}, the joint configuration maintains $\dim(\mathbf{U}_n) > 20$ even at $K = 31$ where individual arrays exhaust their noise subspaces, confirming that spatial stacking genuinely expands the geometric DoFs. The green shaded region ($K \in [32, 57]$) highlights the 26-dimensional capacity gain achievable only through configuration diversity---nearly equal to the entire compressed array contribution. The practical saturation around $K \approx 35$ aligns with the manifold conditioning issues discussed in Remark~\ref{rem:theory_practice_gap}: at very high source counts, the augmented steering matrix becomes increasingly ill-conditioned, limiting practical enumeration performance despite the theoretical capacity.

\begin{figure}[!t]
\centering
\includegraphics[width=0.48\textwidth]{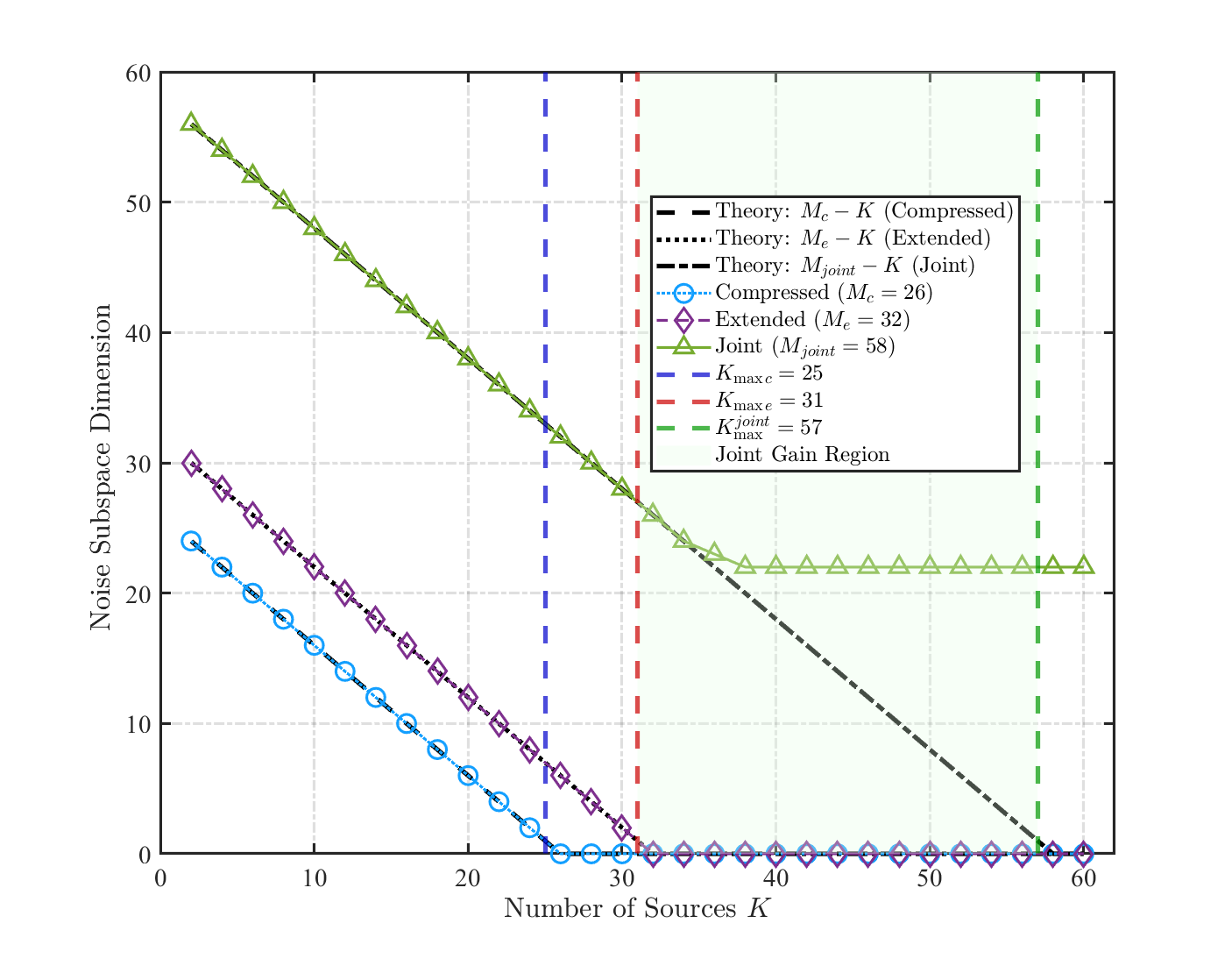}
\caption{Algebraic validation: Noise subspace dimension $\dim(\mathbf{U}_n)$ versus $K$ for compressed, extended, and joint configurations.}
\label{fig:joint_algebraic}
\end{figure}


Fig.~\ref{fig:joint_infotheory} provides the complementary perspective using normalized noise entropy $\bar{H}_{\text{n}}(K) = H_{\text{noise}}(K) / [M_{\text{eff}} \log \sigma_n^2]$. This metric quantifies the fraction of observation entropy attributable to the noise subspace, normalized by array dimension to enable fair cross-configuration comparison. Unlike absolute entropy measures that scale with array size, this normalization reveals the \textit{relative} efficiency with which each configuration uses its available DoFs. The simulation uses SNR = 20~dB and $N = 1000$.

\textbf{Key observations:} (i) \textit{Strict capacity hierarchy}---Joint $\geq$ Extended $\geq$ Compressed for all $K$, with larger arrays maintaining higher normalized noise entropy due to increased observational DoFs. This ordering holds across the entire range, confirming that the metric correctly captures configuration diversity benefits. (ii) \textit{Theoretical bound validation}---All three curves decay smoothly and monotonically to zero precisely at their respective theoretical limits: $K_{\max,c} = 25$, $K_{\max,e} = 31$, and $K_{\max,joint} = 57$. The smooth decay without artificial plateaus or crossovers confirms that identifiability degrades continuously as the noise subspace shrinks, collapsing completely when $\dim(\mathbf{U}_n) \to 0$. (iii) \textit{Joint capacity gain}---The green shaded region ($K \in [31, 57]$) represents a 26-dimensional gain, nearly equal to the entire compressed array contribution ($M_c = 26$), demonstrating that spatial stacking eliminates the sequential bottleneck and fully exploits both configurations' DoFs. This entropy-based metric provides a physically intuitive view: as $K$ approaches $M_{\text{eff}} - 1$, the noise subspace shrinks and observation entropy increasingly reflects signal structure rather than ambient noise, fundamentally limiting parameter discrimination.

\begin{figure}[!t]
\centering
\includegraphics[width=0.52\textwidth]{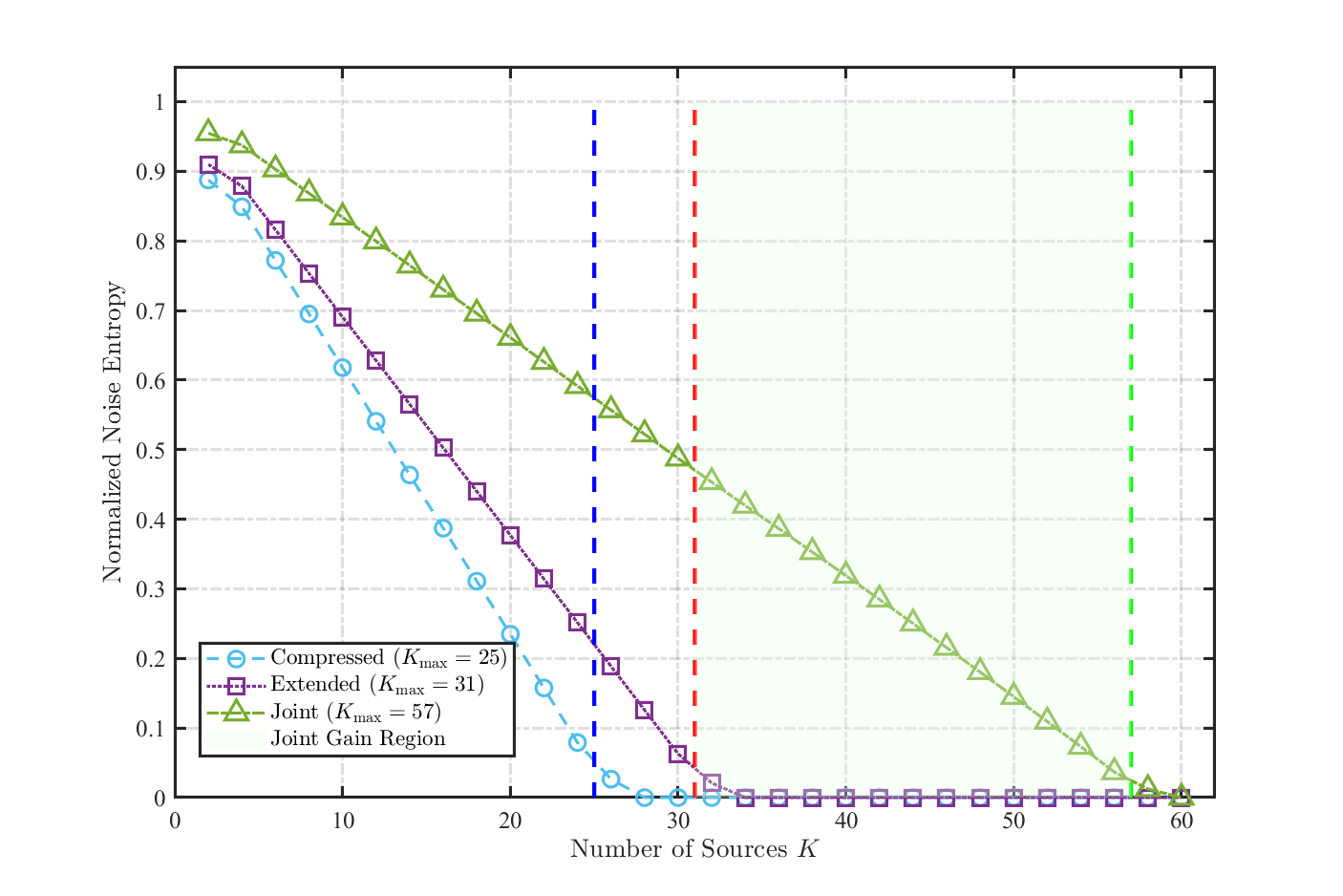}
\caption{Information-theoretic validation: Normalized noise entropy $\bar{H}_{\text{n}}(K)$ versus $K$ for compressed, extended, and joint configurations.}
\label{fig:joint_infotheory}
\end{figure}

\subsection{J-MUSIC Algorithm}

The proposed J-MUSIC algorithm exploits the combined spatial structure across both S-FAS configurations to achieve the theoretical capacity bound $K_{\max}^{\text{joint}} = (M_c+M_e)-1$ derived in Theorem~\ref{thm:joint_dof_identifiability}. The fundamental principle is to construct an augmented observation space by spatially stacking measurements from the compressed and extended configurations, thereby creating an effective $(M_c+M_e)$-dimensional spatial manifold that encodes the source DOAs through augmented steering vectors.

The algorithm begins by forming the augmented observation matrix through spatial concatenation of snapshot vectors from both configurations. For each time instant $n = 1, \ldots, N$, the compressed observation $\mathbf{y}_c(n) \in \mathbb{C}^{M_c}$ and extended observation $\mathbf{y}_e(n) \in \mathbb{C}^{M_e}$ are vertically stacked as $\mathbf{y}_{\text{joint}}(n) = [\mathbf{y}_c^T(n), \mathbf{y}_e^T(n)]^T$, yielding an augmented observation vector in $\mathbb{C}^{M_c + M_e}$. This stacking operation combines the spatial information from both arrays while preserving the distinct geometry of each configuration. The collection of these augmented snapshots forms the observation matrix $\mathbf{Y}_{\text{joint}} \in \mathbb{C}^{(M_c + M_e) \times N}$.

From this augmented observation matrix, the sample covariance matrix is computed as $\mathbf{R}_{\text{joint}} = \frac{1}{N}\mathbf{Y}_{\text{joint}}\mathbf{Y}_{\text{joint}}^H$, which is a $(M_c + M_e) \times (M_c + M_e)$ Hermitian matrix. The EVD yields $\mathbf{R}_{\text{joint}} = \mathbf{U}\mathbf{\Lambda}\mathbf{U}^H$, where the eigenvectors corresponding to the $K$ largest eigenvalues span the signal subspace, while the remaining $(M_c + M_e - K)$ eigenvectors span the noise subspace $\mathbf{U}_n$. The MUSIC exploits the orthogonality between true source steering vectors and the noise subspace.

The spatial spectrum is constructed by evaluating the orthogonality between candidate augmented steering vectors and the noise subspace. For each candidate angle $\theta$ in the search grid, we compute the compressed configuration steering vector $\mathbf{a}_c(\theta)$ and the extended configuration steering vector $\mathbf{a}_e(\theta)$, form the augmented steering vector by stacking: $\mathbf{a}_{\text{joint}}(\theta) = [\mathbf{a}_c^T(\theta), \mathbf{a}_e^T(\theta)]^T$. The MUSIC pseudospectrum is given by
\begin{equation}\label{eq:joint_music_spectrum}
P_{\text{joint}}(\theta) = \frac{\mathbf{a}_{\text{joint}}^H(\theta)\mathbf{a}_{\text{joint}}(\theta)}{\mathbf{a}_{\text{joint}}^H(\theta)\mathbf{U}_n\mathbf{U}_n^H\mathbf{a}_{\text{joint}}(\theta)}.
\end{equation}
At the true source DOAs, the denominator approaches zero (in the absence of noise and model errors), causing the spectrum to exhibit sharp peaks. The DOA estimates are obtained by identifying the $K$ largest peaks in this pseudospectrum. The complete procedure is summarized in Algorithm~\ref{alg:tensor_music}.


\begin{algorithm}[t]
\caption{J-MUSIC for S-FAS}
\label{alg:tensor_music}
\begin{algorithmic}[1]
\State \textbf{Input:} $\mathbf{Y}_c \in \mathbb{C}^{M_c \times N}$, $\mathbf{Y}_e \in \mathbb{C}^{M_e \times N}$, $K$
\State \textbf{Output:} $\hat{\boldsymbol{\theta}}$
\For{$n = 1, \ldots, N$}
    \State $\mathbf{y}_{\text{joint}}(n) = [\mathbf{y}_c^T(n), \mathbf{y}_e^T(n)]^T$
\EndFor
\State $\mathbf{R}_{\text{joint}} = \frac{1}{N}\mathbf{Y}_{\text{joint}}\mathbf{Y}_{\text{joint}}^H$
\State $[\mathbf{U}, \mathbf{\Lambda}] = \text{eig}(\mathbf{R}_{\text{joint}})$
\State $\mathbf{U}_n = [\mathbf{u}_{K+1}, \ldots, \mathbf{u}_{M_c + M_e}]$
\For{each $\theta$ in search grid}
    \State $\mathbf{a}_{\text{joint}}(\theta) = [\mathbf{a}_c^T(\theta), \mathbf{a}_e^T(\theta)]^T$
    \State $P_{\text{joint}}(\theta) = \frac{\|\mathbf{a}_{\text{joint}}(\theta)\|^2}{\mathbf{a}_{\text{joint}}^H(\theta)\mathbf{U}_n\mathbf{U}_n^H\mathbf{a}_{\text{joint}}(\theta)}$
\EndFor
\State $\hat{\boldsymbol{\theta}} = \text{FindPeaks}(P_{\text{joint}}, K)$ \Comment{Select $K$ largest peaks of $P_{\text{joint}}$}
\end{algorithmic}
\end{algorithm}

\subsection{Complexity Analysis}

The computational complexity of Algorithm~\ref{alg:tensor_music} is analyzed as follows. Step~1 constructs the augmented observation matrix through $N$ spatial concatenations, each requiring $M_c + M_e$ memory operations, yielding $\mathcal{O}((M_c + M_e) N)$ operations. Step~2 forms the $(M_c + M_e) \times (M_c + M_e)$ sample covariance matrix via matrix multiplication $\mathbf{Y}_{\text{joint}}\mathbf{Y}_{\text{joint}}^H$, incurring $\mathcal{O}((M_c + M_e)^2 N)$ FLOPs. Step~3 performs EVD of this augmented covariance matrix, which dominates the overall cost with $\mathcal{O}((M_c + M_e)^3)$ FLOPs using standard algorithms such as QR iteration. Step~4 evaluates the MUSIC pseudospectrum over $G_\theta$ angular grid points (the number of candidate DOAs in the search grid), with each evaluation requiring computation of the augmented steering vector ($\mathcal{O}(M_c + M_e)$ FLOPs) and the quadratic form $\mathbf{a}_{\text{joint}}^H\mathbf{U}_n\mathbf{U}_n^H\mathbf{a}_{\text{joint}}$ ($\mathcal{O}((M_c + M_e)^2)$ FLOPs for dense implementation or $\mathcal{O}((M_c + M_e)(M_c + M_e - K))$ for optimized implementation), yielding a total of $\mathcal{O}(G_\theta (M_c + M_e)^2)$ FLOPs for this step. Step~5 involves peak detection over the spectrum, requiring $\mathcal{O}(G_\theta \log G_\theta)$ comparisons. The overall complexity is dominated by Step~3, scaling as $\mathcal{O}((M_c + M_e)^3) = \mathcal{O}((2M-2p)^3)$ for typical S-FAS configurations where $M_c = M - 2p$ and $M_e = M$.


\subsection{Performance Validation}

To validate the J-MUSIC algorithm and quantify its performance advantage, we compare it against conventional single-array methods across three challenging scenarios: varying SNR, varying source count, and varying snapshot number. System settings are: $M=48$ total elements, $M_c=44$ compressed elements at spacing $d_c=0.25\lambda$, and $M_e=32$ extended elements at spacing $d_e=0.5\lambda$, yielding $M_{\text{joint}}=76$ effective sensors.

\subsubsection{Low-Snapshot Regime}

Fig.~\ref{fig:doa_rmse_snr} examines performance with limited snapshots ($N=50$) and three closely-spaced sources ($\Delta\theta \approx 4.5^\circ$) for SNR $\in [-10, 10]$\,dB. J-MUSIC achieves consistent 15--25\% RMSE reduction over Extended MUSIC (red) across all tested SNR levels. The compressed array (MUSIC-C) exhibits higher RMSE due to its smaller effective aperture ($D_c = (M_c-1)d_c$ vs $D_e = (M_e-1)d_e$) despite having more elements ($M_c = 44$ vs $M_e = 32$), while J-MUSIC effectively combines complementary spatial information from both configurations. All algorithms approach their respective Cramér-Rao bounds (CRBs) at moderate-to-high SNR ($\geq 0$~dB), confirming asymptotic efficiency. At very low SNR ($< -5$~dB), finite-sample effects dominate and all methods deviate from their CRBs.

\begin{figure}[!t]
\centering
\includegraphics[width=0.48\textwidth]{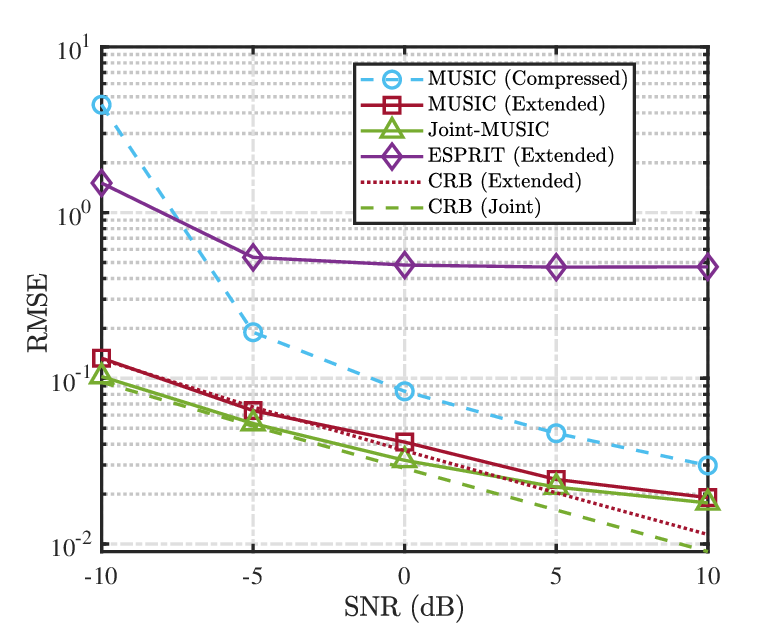}
\caption{DOA estimation: RMSE versus SNR for $K=3$ sources with $N=50$ snapshots.}
\label{fig:doa_rmse_snr}
\end{figure}

\subsubsection{High Source-Count Robustness}

Fig.~\ref{fig:doa_rmse_k} tests scalability under moderate SNR (5\,dB) and limited snapshots ($N=50$) as the number of sources increases from $K=3$ to $K=20$. The compressed array exhibits graceful degradation up to $K \approx 16$, then fails catastrophically when $K \geq 17$ (RMSE $>$ 3$^\circ$) as it exhausts its effective DoFs ($M_c = 44$). Extended MUSIC maintains lower RMSE but also degrades steadily as $K$ approaches its capacity limit. In contrast, J-MUSIC maintains stable performance up to $K=20$ sources with consistent 20--30\% RMSE advantage over Extended MUSIC, demonstrating superior robustness enabled by the augmented 76-dimensional virtual aperture ($M_{\text{joint}} = M_c + M_e = 76$). The practical benefit is clear: J-MUSIC can reliably resolve 3--4 additional sources compared to Extended MUSIC under challenging finite-sample conditions.

\begin{figure}[!t]
\centering
\includegraphics[width=0.48\textwidth]{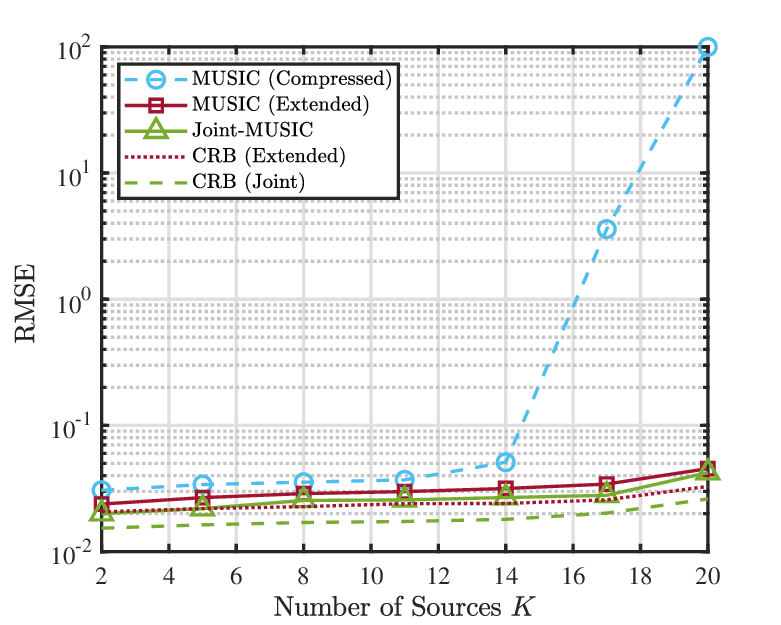}
\caption{DOA estimation: RMSE versus number of sources $K$ at SNR=5\,dB with $N=50$ snapshots.}
\label{fig:doa_rmse_k}
\end{figure}

\subsubsection{Snapshot Efficiency}

Fig.~\ref{fig:doa_rmse_n} evaluates convergence behavior under challenging low-SNR conditions (0\,dB) for $K=3$ sources as snapshot count varies from $N=10$ to $N=1000$. J-MUSIC exhibits faster convergence with increasing snapshots, requiring approximately 40\% fewer samples than Extended MUSIC to achieve equivalent RMSE. For example, J-MUSIC reaches RMSE $\approx 0.1^\circ$ at $N \approx 100$ snapshots, while Extended MUSIC requires $N \approx 160$ for the same accuracy. The compressed array converges more slowly due to limited aperture. At very high snapshot counts ($N \geq 500$), RMSE plateaus at $\sim$0.02$^\circ$ for all methods due to finite search grid resolution (0.05$^\circ$), while theoretical CRBs continue decreasing as $1/\sqrt{N}$. This plateau represents an algorithmic implementation floor rather than a fundamental statistical limit, and can be lowered by using finer angular grids at the cost of increased computational complexity.

\begin{figure}[!t]
\centering
\includegraphics[width=0.48\textwidth]{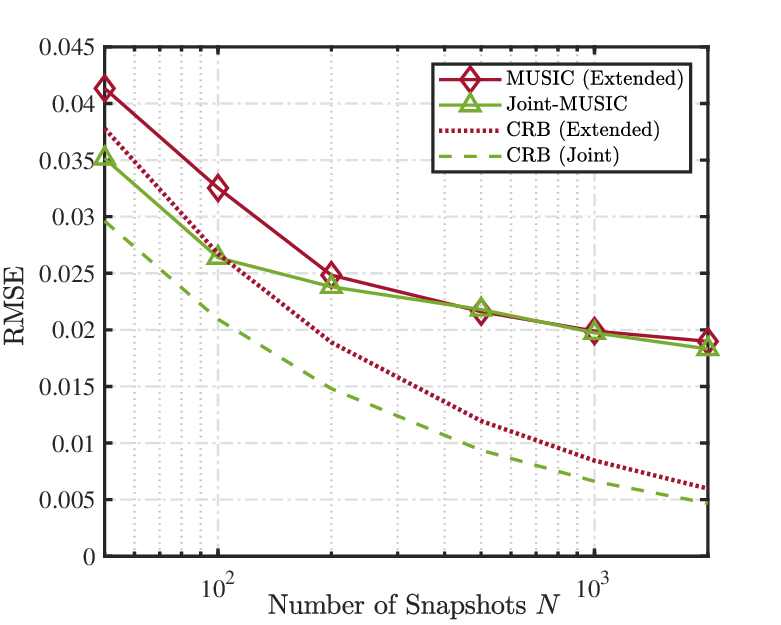}
\caption{DOA estimation: RMSE versus snapshot number $N$ for $K=3$ sources at SNR=0\,dB.}
\label{fig:doa_rmse_n}
\end{figure}

These results validate that J-MUSIC processing yields substantial and consistent DOA estimation improvements (15--30\% RMSE reduction) across diverse operational conditions, while maintaining computational complexity scaling as $\mathcal{O}((M_c+M_e)^3)$.

\section{Conclusion}

This paper developed an observation entropy framework for S-FAS that provides a unified information-theoretic foundation for deriving identifiability limits across all configurations. The central insight is that S-FAS's reconfigurable aperture creates configuration-dependent entropy budgets $H(\mathbf{Y}_\alpha) \leq M_{\text{eff},\alpha} \log(2\pi e \sigma_y^2)$, and identifiability is governed by whether the noise subspace retains sufficient entropy for parameter discrimination ($M_{\text{eff}} - K \geq 1$). From this framework, we derived a complete capacity hierarchy: $K_{\max,c} = M - 2p - 1$ (compressed), $K_{\max,e} = M - 1$ (extended, for both far-field and mixed-field), and $K_{\max}^{\text{joint}} = (M_c + M_e) - 1$ (joint processing). Beyond establishing these bounds, the entropy framework provided three capabilities unavailable from algebraic analysis alone: the data processing inequality diagnosed the sequential bottleneck mechanism, the noise entropy ratio enabled distinction between fundamental DoFs exhaustion and algorithmic suboptimality, and the entropy expansion principle justified joint processing as a means of breaking through single-configuration limits. Monte Carlo simulations with dual algebraic and information-theoretic validation confirmed the predicted boundary behavior and capacity hierarchy across all configurations.

\bibliographystyle{IEEEtran}

\end{document}